\documentclass{ifacconf}

\usepackage[utf8]{inputenc}
\usepackage{amsfonts}
\usepackage{booktabs}
\usepackage{mathrsfs}
\usepackage{amsmath} % assumes amsmath package installed
\usepackage{amssymb}  % assumes amsmath package installed
\usepackage{natbib}            % you should have natbib.sty
\usepackage{graphicx}          % Include this line if your
\usepackage{epstopdf}

\usepackage[dvips]{epsfig}    % or this line, depending on which
\newenvironment{proof}{\mbox{{\em Proof: }}}{ \hfill $\Diamond$ \\}
\newtheorem{theorem}{Theorem}
\newtheorem{example}{Example}
\newtheorem{definition}[theorem]{Definition}

\newtheorem{lemma}[theorem]{Lemma}
\newtheorem{remark}[theorem]{Remark}

\newcommand{\E}{{\mathbb E}}
\newcommand{\R}{{\mathbb R}}

\newcommand{\N}{{\mathbb N}}
\newcommand{\Z}{{\mathbb Z}}

\begin{document}
\begin{frontmatter}

\title{Regularization and Bayesian Learning in  Dynamical Systems: Past, Present and Future} 

\thanks{This work has been partially supported by the FIRB project ``Learning meets time'' (RBFR12M3AC) funded by MIUR.}

\author[First]{A. Chiuso}

\address[First]{Dept. of Information  Engineering, University of Padova (e-mail: \{chiuso\}@dei.unipd.it)}

\begin{abstract}                % Abstract of not more than 250 words.
Regularization and Bayesian methods for system identification have been repopularized in the recent years, and proved to be competitive w.r.t. classical parametric approaches.
In this paper we shall make an attempt to illustrate how the use of regularization  in system identification has evolved over the years, starting from the early contributions both in the Automatic Control as well as Econometrics and Statistics literature. In particular we shall discuss some fundamental issues such as compound estimation problems and exchangeability which play and important role in regularization and Bayesian approaches, as also illustrated in early publications in Statistics.  The historical and foundational issues will be given more emphasis (and space), at the expense of the more recent developments which are only briefly discussed. The main reason for such a choice is that, while the recent literature is readily available, and surveys have already been published on the subject, in the author's opinion a  clear link with past work had not been  completely clarified.
 
\end{abstract}


\end{frontmatter}
%===============================================================================

\vspace{-2mm}
\section{Introduction} \label{intro}
\vspace{-2mm}
About sixty years have passed since the seminal paper by \cite{Zadeh1956}, which has coined the name ``identification'' and  triggered research of the Automatic Control community in the broad area of data based dynamical modeling. 

In the subsequent ten  years the field had achieved such an importance that \emph{IFAC} decided to start the series of \emph{Symposiums on System Identification} (formerly \emph{Symp. on Identification in Automatic Control Systems}) in 1967, two years after the work by \cite{AB1965} which has laid the foundations of Maximum Likelihood methods (and thus  Prediction Error Methods in the Gaussian case) for ARMAX models. I have had the honor and privilege of being a plenary speaker at the  17th Symposium of the series in Bejing, and this paper has been written as a companion to the plenary presentation, which of course could not enter into many of the details that can be found here. 

Despite this long history, the field is still lively and active. We believe there are two main reasons why this is so:
The first is definitely  the increasing importance that data centric methods are playing in many areas of Engineering and Applied Sciences with new challenges arising from the need to process  high dimensional data, possibly in real time, and with little (if any) human supervision. A very nice overview of such challenges has been discussed by Lennart Ljung in his plenary at the joint IEEE Conference on Decision and Control and European Control Conference in 2011 (\cite{LjungCDC2011}) and by Mario Sznaier in his SYSID 2012 semi plenary lecture,  \cite{Sznaier_SYSID_2012}.

The second reason, which is the topic of this paper, has to do with the revitalization of (old) techniques which are rooted in the theory of regularization and Bayesian statistics. 

This paper shall be focused on the role these latter techniques play in the recent developments  of  (linear) system identification, with the main objective to   guide the reader from the early developments to the present days, with a (brief) outlook into the future. For reasons of space we will only address linear system identification, even though we believe the methods and tools discussed here have high potential in the nonlinear scenario as well (see \cite{Tor,NLSVM_2001,PilTAC2011} and references therein). We also warn the reader that only the discrete time case will be presented. It is worth stressing that most of the recent results, just briefly discussed in Sections   \ref{History_NP} and \ref{KernelDesign}, can be framed in a continuous time scenario (see e.g. \cite{SS2010}) so that also non-uniform sampling can be handled (see e.g. \cite{NeveDNM2008} for applications with  pharmacokinetic data).

More specifically, after having introduced the problem and defined notation in Section \ref{sec:probform}, we shall provide  in Section \ref{PEM}
 an overview of parametric Maximum Likelihood/Prediction Error Methods (PEM) as formulated in \cite{AB1965}.  We warn the reader that \emph{this is not} a paper about ML/PEM and thus no attempt is made to discuss its developments over the years. The main goal of Section \ref{PEM} will be  just to set the notation and to pinpoint the weaknesses  of the parametric approach. Section \ref{History_NP} will introduce the regularization approach, with an attempt to provide a complete historical overview, including early work in Statistics and Econometrics where these type of approaches have been first advocated. In order to understand the basic ideas and motivations for the Bayesian approach, Section \ref{CompoundEstimation} introduces the related problem of compound estimation and recalls the notion of exchangeability, as a prerequisite to Section \ref{CompoundEstimationID}
 where their role in the system identification problem will be discussed. In particular the theory of compound estimation provides, from a classical (read frequentist) perspective, a sound theoretical motivation for adopting a regularization/Bayesian point of view.  
 Finally an overview of some recent research in which I have been personally involved,  regarding the design of priors and their use in structure selection problems, will be given in Sections \ref{KernelDesign} and \ref{SparseLowRank} respectively. 
 A brief  outlook into the future will be provided in Section \ref{Outlook}. 

Of course this overview reflects the author's view and other people would have certainly provided a different one. Despite the long list of references I certainly have omitted many relevant ones; yet I still hope this paper can provide a starting point for anybody interested in digging a bit deeper into the roots of Bayesian Learning in System Identification. 

\vspace{-2mm}
\section{Statement of the problem}\label{sec:probform} 
\vspace{-2mm}

Let   $u({t}) \in \R^m$, $y(t) \in \R^p$ be, respectively,  the measured   \emph{input} and \emph{output}  signals in  a dynamical system; the purpose of system identification is to find, from a finite collection of input-output data $\{u(t),y(t)\}_{t\in [1,N]}$,  a ``good'' dynamical model which describes the phenomenon under observation. The candidate model will be searched for within a so-called ``model set'' denoted by   ${\cal M}$. 
%This set can be described in parametric form \cite{Ljung,Soderstrom} or in a nonparamtric form. 
In this paper we shall use the symbol ${\cal M}_n(\theta)$ for parametric model classes where the subscript $n$ denotes the model complexity, i.e. the number of free parameters.

In this paper we shall be concerned with  identification of linear models for jointly stationary processes $\{y(t),u(t)\}_{t\in\Z}$, i.e. models described by  a convolution 
\begin{equation}\label{LDM}
y(t) = \sum_{k=1}^{\infty} g_{k} u(t-k) +  \sum_{k=0}^{\infty} h_{k} e(t-k) \quad t\in \Z.
\end{equation}
where $g$ and $h$ are the so-called impulse responses of the system and $\{e_t\}_{t\in \Z}$ is a zero mean white noise process  which under suitable assumptions is the one-step-ahead prediction error; a convenient description of the linear system \eqref{LDM} is given in terms of the transfer functions
$$
G(q): =  \sum_{k=1}^\infty g_k q^{-k}  \quad  H(q): =  \sum_{k=0}^\infty h_k q^{-k}  \quad 
$$
The linear model \eqref{LDM}  yields an ``optimal'' (in the mean square sense) output predictor which shall be denoted later on by $\hat y(t|t-1)$. As mentioned above, under suitable assumptions, the noise $e(t)$  in \eqref{LDM} is the so-called \emph{innovation} process  $ e(t)= y(t) - \hat y(t|t-1)$. 

In order to simplify the exposition, in this work we shall only deal with feedback free (i.e. assuming that there is no feedback from $y$ to $u$, see \cite{Granger}) Output Error (OE) systems;  thus $H(q)=I_p$ will be postulated. All ideas can be extended to handling, without major difficulties (see e.g. \cite{SS2011}), more general situation involving colored noise (i.e. $H(q) \neq I_p$) as well as the case where feedback is present. This however would obscure the presentation and is thus omitted.

Therefore our focus will be on linear models of the form 
\begin{equation}\label{LDM_OE}
\begin{array}{rcl}
y(t) &= &\sum_{k=1}^{\infty} g_{k} u(t-k) +  e(t) \\ %\quad t\in \Z \quad G(q): =  \sum_{k=1}^\infty g_k q^{-k} \\
&=& \hat y(t|t-1) + e(t)
\end{array}
\end{equation}
where (second order) joint stationarity of $\{y(t),u(t)\}_{t\in\Z}$ implies that $G(q)$ has to be BIBO stable (i.e.  analytic outside the open unit disc of the complex plane, $|q|\geq 1$).
In turn BIBO stability requires  that $g_k$ decays to zero as $k\rightarrow \infty$, and therefore  the infinite summation in \eqref{LDM_OE} can be approximated 
by a finite summation 
\begin{equation}\label{LDM_OE_FIR}
\begin{array}{rcl}
y(t) &\simeq &\sum_{k=1}^{T} g_{k} u(t-k) +  e(t)% \\ %\quad t\in \Z \quad G(q): =  \sum_{k=1}^\infty g_k q^{-k} \\
%&=& \hat y(t|t-1) + e(t)
\end{array}
\end{equation}
for a large enough integer $T$. Since this is always possible up to an arbitrarily small approximation error\footnote{Rigorously one should account for the transient effect, which can be beneficial  for small data sets where $N$ (and thus implicitly $T$) is necessarily small. This can be done rather easily estimating the free response for each output channel. Details are outside the scope of this paper and shall not be discussed here. }, in this paper we shall always work with   FIR models,   assuming exact equality is satisfied in \eqref{LDM_OE_FIR}. This transforms the infinite dimensional model  \eqref{LDM_OE} into a finite dimensional one \eqref{LDM_OE_FIR}.
All the results in this  paper could indeed be formulated with reference to the infinite dimensional model \eqref{LDM_OE}, at the price of bringing so called Reproducing Kernel Hilbert Spaces (RKHS) \cite{Aronszajn50,Saitoh88} into the picture. In our opinion this only entails additional difficulties for the reader and essentially no gain in terms of tools and results and will thus be avoided. The interested reader is referred to \cite{SS2010,SS2011,SurveyKBsysid}. We would like to remind the reader that a detailed study of the asymptotic properties of Bayes procedures for infinite dimensional models is delicate and outside the scope of this paper, see for instance \cite{knapik2011}.

In the following we shall use the notation $Y\in \R^{pN}$ to denote the (column) vector containing the stacked outputs $y(t)$, $t\in[1,N]$ and $\hat Y(g)$ the vector of stacked predictors
$\hat y(t|t-1)$, $t\in[1,N]$ which is a linear function of the impulse response coefficients $g_k$:
$$
\hat Y(g) = \Phi g
$$
where  $\Phi \in \R^{pN\times pmT}$ is a suitable matrix built with the input data $u(t)$, while the column vector $g\in \R^{pmT}$ contains the (vectorized) impulse response matrix coefficient  $g_k\in \R^{p\times m}$, $k\in[1,T]$. 
%\begin{equation}\label{Phi}
%\Phi:=\left[ \begin{array}{ccccc}
%u_0 & u_{-1}& \cdots & u_{-T+1} &\dots  \\
%u_{1}& u_0 & \cdots & u_{-T+2} &\dots  \\
%\vdots & \vdots & \vdots &\ddots & \vdots \\
%u_{N-1}& u_{N-2} & \cdots & u_{N-T} &\dots  
%\end{array}\right]
%\end{equation}
%is a semi-infinite Toeplitz matrix containing the input samples and $g:=[g_1^\top, g_2^\top, ..., g_T^\top,...]^\top$ is the semi-infinite vector with the impulse response coefficients  

Performing identification of $g$ (i.e. estimation of the impulse response from a finite set of input output data) can be thus framed as estimation of the unknown $g$ in the linear model
\begin{equation}\label{LDM_OE_FIR_LM}
\begin{array}{rcl}
Y &= &\Phi g +  E  \quad g \in \R ^d  \quad d:=pmT
\end{array}
\end{equation}

Unfortunately the dimension $d$ of the unknown vector $g$ may be very large (and possibly much larger that the length of the available data $Y$) so that the inverse problem of determining $g$ from $Y$ in \eqref{LDM_OE_FIR_LM}, e.g. minimizing the square loss
\begin{equation}\label{gLS}
\hat g_{LS}: = \mathop{\rm arg\;min}_g \; \|Y-\Phi g\|^2,
\end{equation}
 may be (very) ill conditioned. This may be due to the fact that the input process lives in a high dimensional space ($pm$ large, so that many impulse responses need to be estimated) or simply because $g_k$ decays very slowly to zero and thus many lags need to be included in the parameter vector $g$.  

To face this problem one has to impose constraints on the structure of the vector $g$. One possibility is to parametrize  $g_k = g_k(\theta)$ using a vector $\theta \in \R^n$, $n<<pN$. In the remaining part of the paper we shall sometimes make explicit the dimension of the parameter vector using the notation $\theta_n$.
For instance one may assume  $G(q,\theta):={\cal Z}[g_k(\theta)](q)$ is a rational function of a given McMillan degree  or  $G(q,\theta)$ is expressed via an  orthonormal basis expansion (\cite{WahlbergTAC1991}). When parametric models are considered the notation $\hat y_\theta(t|t-1)$ will be used for the predictor. %Classical parametric approaches \cite{Ljung,Soderstrom} postulate that the transfer functions   $G(q)$ and $H(q)$ belong to a certain parametric class, described by a parameter %vector $\theta\in \R^n$; in such case we shall use the notation $G(q,\theta)$, $H(q,\theta)$.  
We shall come back to these representations later in the paper. 

Note however that even  ``parametric'' representations may require a large number  (possibly infinite) of parameters, such as orthogonal basis expansions.  For instance if  the Fourier  basis $q^{-k}$, $k=1,..,m$ is chosen, a finite expansion has the form $G(q,\theta) = \sum_{k=1}^m  \theta_k q^{-k} =  \sum_{k=1}^m  g_k q^{-k} $. Apparently  the ``parameter'' $\theta=g \in \R^{mpT}$ needs to grow in size in order to approximate any transfer matrix $G$.  We shall refer to this latter cases as ``non-parametric'' (even though, strictly speaking, $g=\theta$ is still a parameter) since  the length of the parameter vector $g=\theta$ may need to be arbitrarily large irrespectively on the number of data $N$ to guarantee that the model class $G(q,\theta)$ accurately describes the ``true'' underlying system.

Whether $g$ is described in a parametric fashion or in a non parametric one, the model ``complexity''   (the complexity  of the model class ${\cal M}$  where  the estimator $\hat g$ is searched for) needs to be controlled so that the estimated model behaves ``well'' on future data (e.g. in terms of output simulation). Therefore, given $Y_{test}$ new data \emph{not} used for identification, it is to be chosen so that 
the so- called \emph{Mean Squared Error}
\begin{equation}\label{MSE}
\begin{array}{rcl}
MSE_Y(\hat g) &= &\mathop {\rm lim}_{N\rightarrow \infty} \E \frac{1}{N} \|Y_{test} - \hat Y(\hat g)\|^2 
%\hat Y & = & \Phi \hat g
\end{array}
\end{equation}
is small. In the particular case of Output Error models with white (unit variance) noise input, the Mean Squared Error \eqref{MSE} reduces to
\begin{equation}\label{MSE}
MSE_Y(\hat g)  = \E\|g-\hat g\|^2 + \sigma^2.
\end{equation}
This is nothing but the usual   \emph{bias-variance tradeoff}: the model has to be rich enough to capture the ``true'' data generating mechanism (low bias) but also simple enough to be estimated using the available data with small variability (low variance). The loss in \eqref{MSE} is called a \emph{compound} loss on the (possibly infinite) vector $g$. As we shall see later on this plays a fundamental role in studying the properties of regularized estimators.

\vspace{-2mm}
\section{Parametric methods: The Maximum Likelihood/PEM era}\label{PEM}
\vspace{-2mm}

The nuts and bolts of Prediction Error Methods (PEM) (or Maximum Likelihood when the innovations $e(t)=y(t)-\hat y(t|t-1)$ are Gaussian) 
were laid down in the System Identification community by the seminal paper  \cite{AB1965} considering SISO  ARMAX  models. In the simplified setup we consider in this paper this corresponds to $m=p=1$, and the model class ${\cal M}_n(\theta_n)$ being described in terms of 
 rational transfer functions $G(q,\theta_n)$ of degree $\nu$:
\begin{equation}\label{parametricOE}
\begin{array}{rcl}
G(q,\theta_n) &=& \frac{\sum_{i=0}^\nu b_i q^i}{q^\nu+\sum_{i=0}^{\nu-1}  a_i q^i}\quad  \quad \sigma^2 = Var\{e(t)\}  \\
\mbox{} & & \\
 \theta_n&:=&[a_0,...a_{\nu-1},b_0,...,b_\nu,\sigma^2]^\top \in \R^{n} \quad n=2\nu+2.
\end{array}
\end{equation}
%where the dimension of the parameter vector $\theta$ is  $n=2\nu+2$.

The estimator $\hat \theta_n$ was found in \cite{AB1965} following the maximum likelihood approach, which under the assumption of Gaussian innovation is equivalent  to solving a (non linear) least squares problem. Let $p_{\theta_n}(Y)$ be the likelihood function of the data $Y$ under the $n-$dimensional parameter vector $\theta_n$, we have:
\begin{equation}\label{ML_OE}
\begin{array}{rcl}
\hat \theta_n & := & \mathop{\rm arg\;\;max}_{\theta\in \R^n} \;\;\; p_{\theta}(Y) \\
& = &  \mathop{\rm arg\;\;min}_{\theta\in \R^n} \;\;\; -2 log\left(p_{\theta}(Y) \right) 
\end{array}
\end{equation}

In \cite{AB1965} also attention was paid to numerical procedures for performing likelihood maximization, which exploit the structure of the underlying model. 
The idea of using Maximum Likelihood for estimating  time series models can  be traced back to \cite{Wald1943}, se also \cite{cramer1946mathematical,grenander1950stochastic,Whittle1952}.  The theory of Maximum Likelihood estimation allows to derive asymptotic results such as consistency and asymptotic statistical efficiency under correct specification of the model class $M_n(\theta)$.  \cite{AB1965} also discuss testing procedures for checking the choice of model order, based on (estimates of) the Hessian of the log likelihood 
$L(\theta_n): = -2 \log\left(p_{\theta_n}(Y) \right) $ and on testing whiteness of the residual sequence $e_{\hat\theta_n}(t): = y(t) - \hat y_{\hat\theta_n}(t|t-1)$. 

The  literature on Statistical System Identification based on the Maximum Likelihood/Prediction Error framework has since then been extensively developed in the Statistics, Econometrics and Control literatures and it is fair to say that it has reached by now a reasonable maturity, as testified by the highly cited textbooks \cite{Ljung:99,Soderstrom,BoxJenkins,BrockwellDavis,HannanDeistler}. It is worth here to recall that much  (but possibly not enough) work has been devoted to the problem of selecting the model class (here the ``model complexity'' $n$). The early literature has focused on testing approaches, see  \cite{Quenouille1947}, \cite{Wold1949} and \cite{AB1965}. Penalty based criteria have been first advocated in \cite{Akaike:74}. The general form is 
\begin{equation}\label{OEC}
\hat n := \mathop{\rm arg\;min}_n \;\; L(\hat\theta_n) + c(N)n 
\end{equation}
where $c(N)$ is a (non decreasing) function of the data length. Different choices lead to different criteria: $c(N)=2$ is known as AIC (\cite{Akaike:74}), $c(N) = \log(N)$ leads to BIC developed in  \cite{Schwarz78} and \cite{Rissanen:78};  other alternatives are possible, see e.g. \cite{Hurvich}. The properties of these order estimation criteria and  of the estimators $\hat\theta_{\hat n}$ obtained after model selection has been performed,  called also \emph{Post Model Selection Estimators} (PMSE), have been studied by a number of authors. We shall here mention only two very relevant contributions:  \cite{YangBiometrika2005}  discusses the  relative merits of AIC and BIC type procedures, i.e. consistency and optimal minimax rate of convergence  and their mutual exclusiveness;   \cite{Leeb,leeb2006} show also that  the asymptotic properties of PMSE are   non-uniform w.r.t the ``true'' model which generates the data, thus making the finite sample properties of the PMSE $\hat\theta_{\hat n}$ possibly very different to the predictions of asymptotic theory.  It is also worth mentioning that, for finite data size, the ``true'' order is not necessarily the ``best'' one (e.g. in terms of facing the bias-variance tradeoff). 

We conclude this section by observing that the model selection step is, in our opinion, a rather critical aspect of the identification procedure. This is in fact in line with the findings, amongst others,  of \cite{Leeb,leeb2006} and with the experimental observations in  \cite{SS2010,SS2011,ChenOL12}.

\vspace{-2mm}
\section{Nonparametric methods, regularization and  Bayesian approaches: formulation and historical overview}\label{History_NP}
\vspace{-2mm}

One of the main drawbacks of parametric approaches, as discussed above, is the need to estimate the model structure (and its complexity $n$) prior to computing a point estimate of the parameter $\theta_n$. On the other hand working in a nonparametric scenario
%, either directly on the impulse response coefficients $g$ or through a suitable basis expansion\footnote{We shall see later on that these are just two special cases of the same procedure.}, 
there is no need to perform ``parametric'' model selection; the truncation length $T$ of the impulse response or the number $L$ of coefficients in  a basis expansion are just supposed to be ``large'' so that the ``true'' model is close enough to the model class. However one needs to find alternative ways to control the model complexity such as to make the inverse problem of determining $g$ from $Y$ in \eqref{LDM_OE_FIR_LM} well-posed.

Thus  a different line of work has studied system identification in the context of ill-posed (non-parametric) linear regression models of the form \eqref{LDM_OE_FIR_LM};  
the main idea can be described either in a probabilistic framework (Bayesian formulation) or deterministic one (regularization). 

From the regularization perspective one can consider estimators of the form 
\begin{equation}\label{RegLS}
\hat g^{Reg}_{\eta} = \mathop{\rm arg\; min}_g \; \| Y-\Phi g\|^2 +  J_\eta(g) 
\end{equation}
for some ``penalty'' function $J_\eta(g)\geq 0$, possibly depending on some parameters $\eta$,  which attempts to ``discourage'' certain undesired solutions. 
Alternatively, the same estimator $\hat g_{\eta}$ can be obtained in a Bayesian framework assuming that $E:= Y-\Phi g$ is a zero mean Gaussian vector with variance $\sigma^2 I$ 
and $g$ has a prior
\begin{equation}\label{prior-penalty}
p_\eta(g) \propto e^{-\frac{1}{2}J_\eta(g)}
\end{equation}
so that the Maximum a Posteriori (MAP) estimator
\begin{equation}\label{BayesLS}
 \hat g^{MAP}_{\eta} = \mathop{\rm arg\; max}_g \; p_\eta(g|Y) = \mathop{\rm arg\; max}_g \; p(Y|g)p_\eta(g), 
\end{equation}
equals $\hat g^{Reg}_\eta$ in \eqref{RegLS}. We shall thus use the notation  $ \hat g_{\eta}:= \hat g^{MAP}_{\eta} =\hat g^{Reg}_{\eta}$ hereafter.

As we shall see an interesting case is when  $J_\eta(g)$ is a quadratic form 
\begin{equation}\label{QuadraticJ}
J_\eta(g) = g^\top K_\eta ^{-1} g \quad K_\eta=K_\eta^\top >0
\end{equation}
where  now the matrix $K_\eta$ is possibly parametrized by a vector of parameters (called hyperparameters, see later on) $\eta$. Under this assumption
the density $p_\eta(g)$ in \eqref{prior-penalty} is Gaussian with zero mean and variance $K_\eta$ and the MAP estimator  \eqref{BayesLS} equals the posterior mean 
\begin{equation}\label{CM}
\hat g_\eta:=\E_\eta [ g|Y]
\end{equation}
which is to be computed for  fixed $\eta$. The matrix $K_\eta$ is also known in the Machine Learning literature as ``kernel'', see \cite{Rasmussen,Scholkopf01b}. Under the assumption \eqref{QuadraticJ} the regularized inverse problem \eqref{RegLS} is known as   Tikhonov-regularization (see \cite{Tikhonov,Melkman1979}) in the literature of inverse problems (see also \cite{Bertero1}) and as ridge regression  in the statistics and numerical analysis literature \cite{Riley:1955:SSL,Foster1961,Hoerl1962,Hoerl:70,Marquardt1975}
\begin{equation}\label{RidgeLS}
\hat g_{\eta} = \mathop{\rm arg\; min}_g \; \| Y-\Phi g\|^2 +  g^\top K_\eta^{-1} g
\end{equation}

The model class is here described as the set of impulse response $g$ such that the term $g^\top K_\eta^{-1} g$ is ``small''. The \emph{size} of this set is related to the singular values profile of $K_\eta$, which are also related to the concepts of \emph{degrees of freedom}, see \cite{Hastie09} and  equation (3.12) in \cite{GCAuto2015}.

Note that the \emph{shape} as well as the \emph{size} of this set depends on the ``regularization (hyper)parameter'' $\eta$, which needs to be fixed in order to  to compute the estimator \eqref{RidgeLS}. This is a way to control the model complexity (the ``size'' of the model class)
which can be done trading \emph{bias} and \emph{variance} controlling the \emph{generalization error} \eqref{MSE}. However, since the Mean Squared Error cannot be computed unless the true system is known, 
estimators have to be found. This can be done by so called \emph{Cross-Validation}, i.e. estimating \eqref{MSE} using data directly (see e.g. \cite{Hastie09}), or by estimators of the MSE such as \emph{Stein Unbiased Risk Estimators} ({SURE}) \cite{SteinAS1981}; also methods based on ``deterministic/worst case'' arguments have been considered in the context of identification, see e.g. \cite{TempoSmoothingParameters1995}.
 
One nice feature of the Bayesian formulation \eqref{BayesLS} as opposed to the regularization one \eqref{RegLS} is that the former provides tools to estimate the hyperparameters $\eta$ using the so-called marginal likelihood  $p_\eta(Y)$, i.e. the likelihood of the hyper parameters $\eta$ once the unmeasurable quantities ($g$) have been integrated out. Let us denote with 
\begin{equation}\label{ML}
\hat \eta:=\mathop{\rm arg\; \;max}_{\eta  }  \;\; p_\eta(Y)
\end{equation}
the marginal likelihood estimator of $\eta$. Then plugging \eqref{ML} in \eqref{CM} we obtain the so-called empirical Bayes estimator (see \cite{robbins1956,robbins1964,good1965,Maritz:1989,Rasmussen,jmlr2014,PCAuto2015}) of $g$:
\begin{equation}\label{EB}
\hat g_{EB}:=\E_{\hat\eta} [ g|Y]
\end{equation}
which is of course a function of the chosen prior $p_{\eta}(g)$ and thus of the kernel $K_\eta$.

To the best of our knowledge Bayesian methods \eqref{BayesLS} for estimating dynamical systems have first been advocated in \cite{Leamer1972} and \cite{Shiller1973} where an FIR model (called \emph{distributed lag} in those references) of the form \eqref{LDM_OE_FIR} was considered. In fact \cite{Leamer1972} and \cite{Shiller1973} were the first to talk about (and apply) ``Bayesian'' methods for system identification, arguing that  ``rigid parametric'' structures may be inadequate to the purpose; the reader may consult instead   \cite{Tiao1964}  for an early overview on the use of Bayes priors in the context of linear regression. These ideas can be traced back, to the best of the author's knowledge, to an  early  paper  by \cite{Whittaker1922}, where smoothing techniques (actually Bayesian priors) where first advocated for the purpose of ``denoising'' some measured data, with particular reference to the field of Actuarial Science in this  early reference. 
 
  In \cite{Shiller1973} it is  assumed that $g$ is a normal random variable with zero mean and covariance matrix $K$, so that the posterior mean of $g$ (conditioned on the observations $Y$) is
\begin{equation}\label{RidgeShiller}
\begin{array}{rcl}
\hat g &=& \E[g|Y] = cov(g,Y)\left[Var(Y)\right]^{-1} Y\\
& = &  K\left(\Phi K \Phi^\top + \sigma^2 I \right)^{-1} Y  = \left(\Phi^\top \Phi + \sigma^2 K^{-1}\right)^{-1} \Phi^\top Y; 
\end{array}
\end{equation}
The latter  (see the rightmost equation) has the form of a (generalized) ridge regression estimator, obtained from \eqref{RidgeLS} with $K_\eta=\frac{1}{\sigma^2} K$.
%$J(g):=\sigma^2 g^\top K^{-1}g$, resulting in 
%\begin{equation}\label{RidgeLS}
%\hat g = \mathop{\rm arg\; min}_g \; \| Y-\Phi g\|^2 + \sigma^2 g^\top K^{-1} g,
%\end{equation}
This is  a well known technique to ``stabilize'' ill-conditioned linear inverse problems (see \cite{Riley:1955:SSL,Foster1961,Hoerl1962,Hoerl:70,Marquardt1975}).

It is clear from equation \eqref{RidgeShiller} that the prior covariance $K$ plays the role of the ``ridge'' matrix in \eqref{RidgeLS}, showing the equivalence between Bayesian estimation in the Gaussian framework and Ridge Regression.

The ridge matrix $K$ was chosen in \cite{Shiller1973} so as to guarantee some degree of smoothness of $g$. However no prior information was enforced on the fact that $g$ is expected to be an (exponentially) decaying sequence. 

An important step in this direction was then made in \cite{Akaike1979} who criticized the choice made by \cite{Shiller1973} of describing the unknown impulse response imposing a smoothness prior. Instead  \cite{Akaike1979} suggests that smoothness should be enforced on the frequency response 
$$G(e^{j\omega}):= \sum_{k=1}^{T} g_k e^{j\omega k}
$$
for instance via the 
  integral of the (squared) first derivative $\frac{dG(e^{j\omega})}{d\omega}$
$$
\frac{1}{2\pi} \int_{0}^{2\pi}\left|\frac{dG(e^{j\omega})}{d\omega}\right|^2 \, d\omega = \sum_{k=1}^T  k^2 |g_k|^2.
$$
Using this as a penalty term  leads to the regularized inverse problem: 
\begin{equation}\label{RR_Bayes_Akaike}
\hat g: = \mathop{\rm arg\; min}_g \;\; \| Y - \Phi g \|^2 + \gamma \sum_{k=1}^T  k^2 |g_k|^2,
\end{equation}
again in the form of a ridge regression estimator. As discussed above  this is completely equivalent to assuming a zero mean Gaussian prior on $g$ as in \cite{Shiller1973}, yet with a decaying diagonal covariance matrix
\begin{equation}\label{KernelAkaike}
K_\gamma:= \frac{1}{\gamma} {\rm diag}\left\{1, \;\frac{1}{4}, \; \frac{1}{9},\dots,  \frac{1}{T^2} \right\}
\end{equation}

It is also worth observing that \cite{Akaike1979} is the first to advocate the use of the Empirical Bayes  approach (probably first introduced in \cite{robbins1956,robbins1964} and \cite{good1965} under the name of \emph{Type II maximum Likelihood}) to estimate  the scaling factor $\gamma$. More specifically $\gamma$ is estimated as\footnote{Slight variations suggested by Akaike include marginalization with respect to the noise variance $\sigma^2$, subject to improper priors.}:  
\begin{equation}\label{ML_Akaike}
\hat \gamma:=\mathop{\rm arg\;max}_{\gamma} \int p(Y|g) p(g|\gamma)\, dg
\end{equation}
and then substituted in \eqref{RR_Bayes_Akaike}. Note that with Gaussian errors $E$ (so that $p(Y|g)$ is a Gaussian with  mean $\Phi g$ and variance $\sigma^2 I$) and Gaussian prior $p(g|\gamma)$ 
the integral in \eqref{ML_Akaike} is available in closed form. It is interesting to observe that Akaike  suggested to use the marginal likelihood also to estimate other system properties not captured by the covariance \eqref{KernelAkaike} such as  possible delays.

This very same approach was then followed by a series of works by Kitagawa and Gersh, see e.g. \cite{KitagawaJASA1984,KitagawaTAC1985} which are well documented in  the book \cite{kitagawa1996smoothness}. 

A related formulation, yet with a slightly different goal, is found in the so called ``stochastic embedding'' approach by \cite{Goodwin1989,Goodwin1992}, see \cite{Ljung_SE_2014} for details on a ``modern'' interpretation. Essentially the idea was to admit that the ``true'' transfer function $G(q)$ is only partially captured by the chosen model class ${\cal M}_n(\theta)$
so that 
\begin{equation}\label{model_error}
G(q)  = G(q,\theta_0) + \tilde G(q) \quad G(q,\theta_0) \in {\cal M}_n(\theta)
\end{equation}
and  $\tilde G(q)$ represents a model error. \cite{Goodwin1992} describe this model error using a ``stochastic'' approach similar to \cite{Akaike1979}, yet with an exponentially decaying covariance function. More specifically:
$$
\tilde G(q):=\sum_{k=1}^{T} \tilde g_k q^{-k}
$$
where $\tilde g:=[vec(\tilde g_1)^\top,.. vec(\tilde g_T)^\top]^\top \sim {\cal N}(0, K)$ with 
\begin{equation}\label{KernelSE}
K_{\rho,\lambda}:=\lambda \,{\rm diag}\left\{1,\; \rho,\;\rho^2, \dots, \rho^{T-1} \right\}
\end{equation}  
The model error $\tilde G(q)$ is estimated in  \cite{Goodwin1992} starting from the least squares residuals $v_{\hat\theta}(t):=y(t) - G(q,\hat \theta)u(t)$ where 
$$
\hat \theta = \mathop{\rm arg\;min}_\theta \; \sum_{k=1}^T \| y(t) - G(q,\theta)u(t)\|^2
$$
Under the assumption \eqref{model_error}, the error  $v_{\hat\theta}(t):=y(t) - G(q,\hat \theta)u(t)$  is expected to be described by the model 
$$
v(t) = \tilde G(q) u(t) + e(t) 
$$ and thus 
\cite{Goodwin1992}  propose to estimate $\lambda$ and $\rho$ from marginal likelihood maximization
\begin{equation}\label{ML_SE}
(\hat\lambda, \hat \rho):=\mathop{\rm arg\;max}_{\lambda,\rho} \int p(V_{\hat\theta}|\tilde g) p(\tilde g|\lambda,\rho)\, d\tilde g
\end{equation}
where $V_{\hat\theta}:=[v^\top_{\hat\theta}(1),...,v^\top_{\hat\theta}(N)]^\top$.
 This  leads to an empirical Bayes estimator similar to that utilized in \cite{Akaike1979}. Note however that \cite{Goodwin1992} is not focused on estimating the transfer function
 $\tilde G(q)$ but rather  its variance, which is utilized to describe statistically the model error. This in turn is used  to introduce an order estimation criteria 
 based on minimization of an estimate of the Mean Squared Error  $$
\E \| G(q) - G(q,\hat \theta_n)\|^2
 $$
while accounting for the model error model $\tilde G(q)$ (see e.g. equations (89-92) in \cite{Goodwin1992}). It is worth to observe that the stochastic embedding approach, which postulates a white noise model with exponentially decaying variances  \eqref{KernelSE}
on the impulse response sequence, implies a Lipschitz  condition on the impulse response:
\begin{equation}\label{frequency_domain_smoothness}
\begin{array}{c}
 E[\| G_{{\lambda}}(e^{j\omega_1}) - G_{{\lambda}}(e^{j\omega_2})\|^2] \leq f({\lambda}) (\omega_1-\omega_2)^2\\
 \mbox{}\\
 \displaystyle{\mathop{\rm lim}_{|{\lambda}|\rightarrow 1^-}} f({\lambda}) = +\infty
 \end{array}
\end{equation}
 similar in spirit to the ``frequency domain'' smoothness condition advocated in \cite{Akaike1979}.

A significant amount of work with Bayesian flavor is also found in the econometrics literature. For instance  in a series of works Litterman, Doan and Sims (see e.g. \cite{DoanLSER1984} for an overview) study time varying multivariate autoregressive models of the form\footnote{Exogenous inputs $u(t)$ have also be considered in this framework.}
\begin{equation}\label{MVAR_DLS}
y(t)  = \sum_{k=1}^m A_{k,t} y(t-k) + C(t)+ \epsilon(t)
\end{equation}
where $C(t)$ is a time varying deterministic sequence and $\epsilon(t)$ is a white noise sequence.

The coefficients $A_{k,t}$ and $C(t)$ are modeled as a function of a number of  hyperparameters (called $\pi_i$, see page 12 of  \cite{DoanLSER1984})) as follows:
The time evolution is modeled as
  a first order auto regression 
$$
{\rm vec}\left(A_{k,t+1}\right) = \pi_8 {\rm vec}\left(A_{k,t}\right) + (1- \pi_8) {\rm vec}\bar \mu + \mu(t)
$$
where $\bar\mu$ is a fixed vector and $\mu(t)$ is a known time varying sequence. The initial conditions of this recursion (mean and variances)  are fixed
 (see eq. (4-6) \cite{DoanLSER1984}). In particular the 
coefficients $[A_{k,0}]_{i,j}$ are supposed to satisfy
\begin{equation}\label{VarianceDLS}
\begin{array}{l}
Var\{[A_{k,0}]_{i,j}\} =\ \left\{\begin{array}{rcl}
\frac{\pi_5 \pi_1}{k e^{\pi_4 w_{ii}}} & \mbox{}& i=j\\
\frac{\pi_5 \pi_2\sigma_i^2}{k e^{\pi_4 w_{ij}}\sigma_j^2} & \mbox{}& i\neq j
\end{array} \right.\\
Cov\{A_{k,0}]_{i,j},A_{h,0}]_{m,n}\} =Var\{[A_{k,0}]_{i,j}\} \delta_{hk}\delta_{im}\delta_{jn}
\end{array}
\end{equation}
where $\delta_{hk}$ is the Kronecker delta, $w_{ij}$ are suitable user defined weights and $\sigma_i^2$ are the variances of the prediction error for the $i-$th component of $y(t)$. 
This is one form of the so called \emph{Minnesota prior}, which has been discussed quite extensively in the econometrics literature; several variations and extensions are found, see for instance  \cite{Ltkepohl:2007,Giannone2015}.

It is apparent that the resulting  prior distribution has a considerable flexibility; it is reasonable to expect that the performance of any algorithm based on such prior may heavily depend upon the (rather arbitrary) user choices; 
 \cite{DoanLSER1984} discuss extensively   on the meaning of each of these choices and give suggestion as to how these should be made; delving into these details  is  outside the scope of this lecture. 
Suffices here to notice that, assuming the coefficients are fixed over time, equations \eqref{VarianceDLS} imply a  diagonal prior covariance matrix with  elements decaying as $\frac{1}{k}$ where $k$ is the lag index while different versions of the Minnesota Prior suggest decay rates of the form $\frac{1}{k^\alpha}$; this clearly  reminds  of the prior suggested by \cite{Akaike1979}, in  which $\alpha=2$.  Note that $\alpha>1$ is required in order to guarantee that  realizations $g_k$, $k=1,..,T$ are (a.s.) in $\ell_1$ as $T\rightarrow \infty$. % (note that, instead Akaike's prior does guarantee  $\{g_k\}_{k=1,..,T}$ is  $\ell_2$ as $T\rightarrow \infty$).

The econometrics literature has since then studied Bayesian procedures  for system identification rather intensively, mostly under the acronym \emph{Bayesian VARs}); the main driving motivation was that of handling high dimensional time series (i.e. $p$ large, called \emph{cross sectional dimension} in the econometrics literature) with possibly many explicative variables ($m$ large), see for instance \cite{Knox2001,DeMol2008,BVAR2010,Giannone2015}. While 
\cite{DoanLSER1984} and \cite{BVAR2010} propose tuning the hyperparametrs using out-of-sample and in-sample error respectively, \cite{Knox2001} and the most recent work \cite{Giannone2015} adopt an Empirical Bayes approach  using the marginal likelihood (i.e. the likelihood function of the data $Y$ given the hyperparameters once the unknown parameters have been integrated out)  for hyperparameter estimation;   \cite{Giannone2015}  claims the superiority of this approach w.r.t.  previous ``ad-hoc'' techniques \cite{DoanLSER1984,BVAR2010}.

The most recent developments in the area of system identification, which have followed the seminal contribution \cite{SS2010} and the subsequent papers   \cite{SS2011,ChenOL12,ChiusoPAuto2012} (see also the survey \cite{SurveyKBsysid}) are based on modeling the unknown impulse response $g$ as a zero mean Gaussian process with a suitable covariance matrix  $K_\eta$. Most of the work has been focused on designing the matrix $K_\eta$ (see \cite{SS2010,SS2011,ChiusoPAuto2012,ChenOL12,ChenTAC2014,ChenLOBF2015}) as well as on studying the properties of the Empirical Bayes estimator \eqref{EB}, \cite{jmlr2014,PCAuto2015}. More details will be given later on.

%\begin{equation}\label{KernelNostro}
%
%\end{equation} 
\vspace{-2mm}
\section{Compound decision and Estimation problems and Exchangeability }\label{CompoundEstimation}
\vspace{-2mm}

In this section we shall attempt to bridge the gap between the classical ``frequentist'' view in statistics, which is rooted in the idea of ``repeated'' experiments with fixed but unknown parameters, and the Bayesian view where parameters are assumed to satisfy a prior distribution and are estimated minimizing suitable risk functionals.

Our detour goes back to early work by \cite{robbins1951,robbins1956,robbins1964}, \cite{Stein1956,Stein1961}, \cite{good1965}, \cite{LindleySmith1972}, \cite{EfronMorris1973}. 
 
Two mind blowing papers in statistics have been \cite{robbins1951} and \cite{Stein1956} which studied respectively the so called compound decision and estimation problems. The main message is as follows: assume there is a parameter vector $\alpha \in \R^B$ which we need to make a decision on or estimate. Two possibilities are that either the components 
of $\alpha$ are binary variables (e.g. $\alpha_i \in \{-1,1\}$) (\cite{robbins1951,robbins1956}) or continuous variables $\alpha_i\in \R$ (\cite{Stein1956,Stein1961}). For simplicity of exposition, let us consider the estimation problem and assume measurements are available of the form 
\begin{equation}\label{RobbinsStein}
Y = \alpha + E
\end{equation} where $E$ is a Gaussian noise with zero mean and covariance $\sigma^2I$ (assume for simplicity that $\sigma^2$ is known, but similar arguments hold if $\sigma^2$ needs to be estimated from data). 

Assume  that, given a rule $\hat\alpha:= \delta(Y)$ used to estimate $\alpha$, one measures the loss incurred  in the ``compound'' manner:
\begin{equation}\label{compoundLoss}
L(\alpha,\delta): =\frac{1}{B} \sum_{i=1}^B (\alpha_i-\delta_i(Y))^2
\end{equation}
This means that we are not interested in the error $\alpha_i-\delta_i(Y)$ on each  component, but only on the aggregate error \eqref{compoundLoss} which is symmetric (i.e. permutation invariant).  Note that in the System Identification scenario one is typically interested exactly in compound  losses of this form, see \eqref{MSE}.
We now need to define the concepts of \emph{admissibility}.% and of \emph{minimax} estimators

\begin{definition}
An estimator $\delta$ is \emph{admissible} (w.r.t. the loss function $L(\alpha,\delta)$) if there does not exist any other estimator $\delta^*$ such that 
$$
\E L(\alpha,\delta^*) \leq \E L(\alpha,\delta)  \quad \forall \alpha
$$
with  strict inequality for some $\alpha$. If such an estimator $\delta^*$ exists, $\delta$ is inadmissible and $\delta^*$  is said to dominate $\delta$. 
\end{definition}

%\begin{definition}
%An estimator $\delta^{*}$ is said to me \emph{minimax} (w.r.t. the loss function $L(\alpha,\delta)$) if it minimizes the worst case error:
%$$
%\delta^{*} = \mathop{\rm arg\;min}_\delta \left\{ {\rm inf}_\alpha \E L(\alpha,\delta) \right\}
%$$ 
%\end{definition}

%\begin{remark}
%Since the least squares estimator $\delta^{LS}$ has constant risk (and it is inadmissible)  $\E  L(\alpha,\delta^{LS})$ then a minimax estimator $\delta^{*}$ also dominates $\delta^{LS}$
%\end{remark}
\cite{robbins1951,robbins1956}  has shown that, for the decision problem $\alpha_i\in\{-1,1\}$, the ``classical''  Neyman-Pearson test (\cite{lehmann2005testing})
\begin{equation}\label{NP}
\hat \alpha_i=\delta^{NP}_{i}(Y) = sign(Y_i)
\end{equation} 
is not admissible. Similarly, \cite{Stein1956} has shown that, for the estimation problem $\alpha_i\in \R$,
 the least squares estimator
\begin{equation}\label{LS}
\hat \alpha_{i} =\delta^{LS}_{i}(Y)= Y_i
\end{equation} 
is \emph{inadmissible} (for $B>2$). 
Estimation (or decision) rules which dominate $\delta^{NP}_{i}(Y) $ and $\delta^{LS}_{i}(Y)$  are provided in   \cite{robbins1956} and  \cite{Stein1961} respectively. We shall now study in more detail the case considered in \cite{Stein1956,Stein1961}.

%, meaning that there are other estimation rules $\delta^*(Y)$ such that the expected loss (expectation taken w.r.t. the noise $E$)
%$$
%\E{L(\alpha,\delta(Y))}
%$$ are  not worse (over $\alpha$) than that obtained using \eqref{NP} and \eqref{LS},  while being strictly better for some value of $\alpha$. Sometimes it is said that there are  rules $\delta^*$ which dominate $\delta^{NP}_{i}(Y)$ and $\delta^{LS}_{i}(Y)$. 

The most famous example of such rules which dominate  Least Squares is the so-called %(positive part) 
James-Stein estimator (discovered by James, which was at the time a student of Stein, \cite{Stein1961}):
\begin{equation}\label{JS}
\hat \alpha = \delta^{JS}(Y) = \left(1-\frac{(B-2) \sigma^2}{\|Y\|^2}\right) Y
\end{equation}
%where the positive part operator $(x)^+$ is defined as  $(x)^+_i = max(x_i,0)$.
It is clear that the estimation rule $\delta^{JS}(Y)$ has a ``compound'' flavor, in that the estimator  $\delta_i^{JS}(Y)$ of the $i-th$ component $\alpha_i$ depends on all the measurements $Y$ and not only on $Y_i$. This may be surprising given  that the noise $E$  in \eqref{RobbinsStein} has uncorrelated components. It is an easy calculation to show that 
$$
\E L(\alpha,\delta^{JS}) < \E L(\alpha,\delta^{LS}) = \sigma^2 \quad \forall \;\;\alpha \; : \; \|\alpha\| < \infty $$
provided $B>2$. \cite{EfronMorris1973}  discuss the link with the empirical Bayes approach. We derive here a simplified analysis. 
Let us postulate a  Gaussian prior on $\alpha \sim {\cal N}(0,\lambda I)$ and assume the noise variance $\sigma^2$ is known. The Bayes estimator \eqref{CM} takes the form
\begin{equation}\label{CMStein}
\hat\alpha_\lambda:=\E_\lambda [\alpha|Y] = \frac{\lambda}{\lambda + \sigma^2} Y \end{equation}
and the  hyperparameter $\lambda$ can be estimated using \eqref{ML} leading to 
\begin{equation}\label{MLStein}
\hat\lambda = max\left(\frac{1}{B}\|Y\|^2 - \sigma^2,0\right)
\end{equation}
Thus leading to the Empirical Bayes estimator
\begin{equation}\label{EBStein}
\hat\alpha_{EB}:=\hat\alpha_{\hat\lambda} = \left( 1 - \frac{B \sigma^2}{\|Y\|^2}\right)^+ Y
\end{equation}
where, given $v\in \R^B$, the positive part operator $(v)^+$ is defined as  $(v)^+_i = max(v_i,0)$. The estimator \eqref{EBStein}  has the same form of the positive-part version of \eqref{JS} (see \cite{Stein1961}) with a slightly different numerator ($B$ in lieu of $B-2$), which can shown to dominate the Least Squares estimator \eqref{LS} provided $B$ is large enough.

\begin{remark}
An alternative ``Empirical Bayes'' derivation of the James Stein estimator \eqref{JS} goes as follows:  Instead of maximizing the marginal  posterior distribution of $Y$, one can observe that,  $\|Y\|^2 \sim (\sigma^2+\lambda)\chi^2_B$, so that 
$\hat A:=\sigma^2 \frac{B-2}{\|Y\|^2}$ is an unbiased estimator of $A=\frac{\sigma^2}{\sigma^2 + \lambda}$. Then, replacing $\hat A$ for $A $ in the expression 
$\hat\alpha_\lambda:=\E_\lambda [\alpha|Y] = \frac{\lambda}{\lambda + \sigma^2} Y  = (1-A)Y$ (see \eqref{CMStein}),  
gives:
$$\hat\alpha'_{EB}=\left( 1 -\hat A\right) Y =\left( 1 - \sigma^2 \frac{B-2}{\|Y\|^2}\right) Y = \delta^{JS}(Y)
$$
\end{remark}

The overall message we learn from this example is that ridge regression/Bayesian estimation can provide valuable tools to build estimators which dominate standard (non-regularized) Least Squares provided one is interested in compound losses \eqref{compoundLoss} and not  in the quality of a single parameter  $\alpha_i$. 
In addition the Empirical Bayes paradigm provides also a  tool for selecting, in the class of Bayes estimators \eqref{CMStein}, a specific one which dominates Least Squares. 
The Empirical Bayes rule leads  to the ``compound'' estimator \eqref{EBStein} which couples all observations similarly to \cite{Stein1961}. Unfortunately the situation is rather tricky if one considers measurement models such as  
\begin{equation}\label{generalModel}
Y = S \alpha + E
\end{equation}
and more general (quadratic) loss functions \eqref{compoundLoss} such as:
\begin{equation}\label{compoundLossQ}
L(\alpha,\delta):=\left(\delta-\alpha\right)^\top Q \left(\delta-\alpha\right) 
\end{equation}

The statistics literature is rich of contributions in this direction which is certainly impossible to survey here. The reader is referred to \cite{Strawd1978,Berger1982,Casella1980,Strawd2005} and references therein.  We shall come back to this issue in Section \ref{CompoundEstimationID}
where such connection will be studied in some detail. 

Before concluding this section  there is one more important remark to make which has to do with the concept of exchangeability, see  \cite{deFinetti1931,HewSav55,LindleySmith1972}. 
\begin{definition}
We say that an infinite sequence of random variables $x_i$, $i\in \N$ is \emph{exchangeable} if $\forall h\geq 2$ the joint distribution of $x_1,..,x_h$ is invariant to permutations of the indexes set $1,..,h$.
\end{definition}
An important characterization of exchangeable random sequences is due to \cite{deFinetti1931} and extended in \cite{HewSav55} who have shown that  an infinite sequence of random variables $x_i$, $i=1,..,$ is exchangeable if  and only if there exists a random distribution $\Pi(x)$ such that the joint distribution $P(x_1,..,x_h)$ satisfies
\begin{equation}\label{exchangeablePrior}
P(x_1,..,x_h) = \int \prod_{i=1}^h \Pi(x_i)\,Q(d\pi) 
\end{equation} for a suitable measure $Q(\pi)$. This is to say that, conditionally on the distribution $\Pi$, the $x_i$'s are i.i.d. with distributions $\Pi(x)$. % $\pi(x) = \frac{d\Pi(x)}{dx}$$. 

For the setup considered in \cite{robbins1951} and \cite{Stein1956}, exchangeability of $\alpha_i$ is a reasonable assumption which is also reflected in the loss function $L(\alpha,\delta)$ being insensitive to permutations of $\alpha_i$'s. The result in equation \eqref{exchangeablePrior} essentially says that, for an exchangeable sequence, it is reasonable to assume that there is a hierarchical prior model for the joint distribution. Further, results in \cite{deFinetti1931,HewSav55} suggest that the distribution $\Pi(x)$ can be estimated from a sample $x_1,...x_h$. When the $x_i$'s are not observed, but only their (noisy) measurements $y_i$ are available, then  $\Pi(x)$ is to be estimated from $y_1,...y_h$.  The reader may consult also \cite{Edelman1988} where further discussion concerning the estimation of $\Pi$  as well as theoretical results on the resulting estimator's properties are derived. 

The Empirical Bayes estimator \eqref{EBStein} can be framed in this context, provided one postulates that $\frac{d\Pi}{dx}$ is the density of a zero mean Gaussian random variable with variance $\lambda$. Then $\lambda$ can be estimated from the sample $y_1,...,y_h$ as in \eqref{MLStein}. The next Section will elaborate upon this in some detail.
 
\vspace{-2mm}
\section{Connections between Compound estimation problems and System Identification}\label{CompoundEstimationID}
\vspace{-2mm}

We shall now see how the  system identification problem fits in the framework of compound estimation discussed in the previous Section. As anticipated  at the end of Section 
\ref{CompoundEstimation}, we shall need to consider measurement models of the form \eqref{generalModel} and ``weighted'' losses of the form \eqref{compoundLossQ}.

Consider now the linear model   \eqref{LDM_OE_FIR_LM} and further assume that 
the unknown vector $g$  is expanded w.r.t. a (possibly  orthogonal) basis $\psi_i$, such that $\|\psi_i\|_{\ell^1}  = \mu_i$, $i=1,..,B$: 
\begin{equation}\label{BasisExpansion}
g = \sum_{i=1}^B \psi_i\alpha_i
\end{equation}
For convenience of notation we define $\Psi:=[\psi_1,..,\psi_B    ]$ and $\alpha:=[\alpha_1,...,\alpha_B]^\top$.
With such notation, and  using \eqref{BasisExpansion},  model \eqref{LDM_OE_FIR_LM} can be rewritten in the form 
\begin{equation}\label{LDM_OE_FIR_LM_BE}
Y = \Phi \Psi \alpha + E
\end{equation}
which is of the form \eqref{generalModel} with $S:= \Phi \Psi $.

We shall come back in Section \ref{KernelDesign} on the issue of selecting the basis vectors $\psi_i$. Assume also that the sequence $\mu_i$ measures the ``importance'' of the $i-th$ basis vector $\psi_i$ and is such that, in the limit\footnote{Note that here, in order to take the limit as $B\rightarrow\infty$, we should think of $g$ as an element of $\ell(\Z^+)$.} as $B\rightarrow\infty$, 
\begin{equation}\label{summability}
\sum_{i=1}^\infty | \mu_i| <\infty
\end{equation}
so that the sum \eqref{BasisExpansion} converges (in $\ell_1$) provided $|\alpha_i|<c$, $\forall i$  and converges in mean square if the $\alpha_i$'s are i.i.d. with zero means and finite variance. 

Note that, if one postulates that $\alpha \sim{\cal N}(0,\lambda I)$, then 
\begin{equation}\label{GaussianPrior}
g = \Psi\alpha \sim{\cal N}(0,\lambda K)  \quad \quad K:= \Psi\Psi^\top
\end{equation} where now the columns of $\Psi$ are the (unnormalized) eigenvectors of $K$ and \eqref{BasisExpansion} is the (normalized) Karhunen-Lo\`eve (\cite{Loeve})  expansion of $g$. This leads back to the standard  zero mean Gaussian prior on the unknown $g$ adopted in  the Bayesian literature (\cite{Tiao1964,Leamer1972,Shiller1973,Akaike1979,DoanLSER1984,DeMol2008,Goodwin1992,BVAR2010,SS2010,SS2011,ChenOL12}). 

According to the discussion in the previous section, this assumption  implicitly implies that the vector $\alpha$ is (a sub vector of) an exchangeable sequence. In our personal view this is indeed reasonable since the basis vectors $\psi_i$ capture the properties of $g$ while the coefficients  $\alpha_i$ all have  the same weight  and  play a symmetric role.  Thus permutation  invariance of their distribution seems a reasonable assumption. We shall leave it to the reader to judge whether this is an acceptable assumption to make or not, and indeed we would welcome comments and discussions on this. 

Conversely, from \eqref{exchangeablePrior}, the exchangeability assumption on $\alpha$  only guarantees the existence of a random distribution $\Pi$, which needs not be Gaussian. Thus the hierarchical model  \eqref{GaussianPrior}  entails the additional  assumption on the parametric form of  $\Pi$ (zero mean Gaussian with variance $\lambda$) so that \eqref{exchangeablePrior} leads to  \eqref{GaussianPrior}.

Similarly to \cite{jmlr2014} we now refer to model \eqref{LDM_OE_FIR_LM_BE} and introduce
 the Singular Value Decomposition $QD^{1/2}V^\top =  \Phi \Psi $.  Define also  $Z:= Q^\top Y$,  $\beta:= V^\top \alpha$ and with some abuse of notation denote $Q^\top E$ still as $E$ (since $Q^\top E$ has the same distribution as $E$ when $E\sim {\cal N}(0,\sigma^2 I)$). 

The linear model \eqref{LDM_OE_FIR_LM_BE} becomes
\begin{equation}\label{LDM_OE_FIR_LM_BE_NEW}
Z =D^{1/2}\beta + E = \bar \beta + E \quad \bar \beta:=D^{1/2}\beta 
\end{equation}

A simple calculation shows that the  Bayes  estimator (say $\hat {\bar\beta}=\delta(Z)$) \eqref{CM} of $g$,  in the ``new'' coordinates $\bar\beta =D^{1/2}V^\top  \alpha$, takes the form of a shrinkage estimator; written componentwise: 
\begin{equation}\label{EM_Coordinates}
\delta_i(Z)=\hat{\bar\beta}_i =\left(1-\frac{\sigma^2/d_i}{ \lambda + \sigma^2/d_i} \right) Z_i
\end{equation}

Note that  model \eqref{LDM_OE_FIR_LM_BE_NEW} has the same form as (29) in \cite{jmlr2014} with the difference that now the prior covariance needs not be a multiple of the identity.  It is now interesting  to connect model \eqref{LDM_OE_FIR_LM_BE_NEW} with the change of coordinates suggested in \cite{Strawd1978}, which considers a linear estimation problem of the form \eqref{LDM_OE_FIR_LM}, \eqref{LDM_OE_FIR_LM_BE}. We again here refer to  \cite{Strawd1978}, which introduces a change of coordinates in the parameter $g$ (named $\theta$ in \cite{Strawd1978}), defining $\bar g = A g =A \Psi \alpha$ where $A$ is such that\footnote{Note that one should take  $B$ in \cite{Strawd1978} so that $B^{-1}:=\Phi^\top \Phi K \Phi^\top \Phi$; with this definition of $B$  and   given $A$ such that
$ A(\Phi^\top\Phi)^{-1}A^\top = I$ holds, $AKA^\top = \Lambda$  is equivalent to   $ABA^\top = \Lambda^{-1}$.}, for some diagonal matrix $\Lambda$, 
\begin{equation}\label{Astrawderman}
AKA^\top = \Lambda \quad \quad A(\Phi^\top\Phi)^{-1}A^\top = I
\end{equation}
The structure of this change of coordinates is characterized by the following lemma.
\begin{lemma}\label{LemStraw}
Given the Singular Value Decompositions $\Phi\Psi = QD^{1/2}V^\top$ and $\Phi = PS^{1/2}U^\top$, the change of coordinates $\bar \beta = A g=A \Psi \alpha$ leading to \eqref{LDM_OE_FIR_LM_BE_NEW} is such that \eqref{Astrawderman} is satisfied with $\Lambda  = D$ if
\begin{equation}\label{AS}
A =Q^\top P S^{1/2}U^\top = Q^\top \Phi = D^{1/2} V^\top \Psi^{-1}
\end{equation}
\end{lemma}
\begin{proof}
It is a simple matrix manipulation to check that $A = \bar Q S^{1/2}U^\top$ for some $\bar Q$: $\bar Q^\top \bar Q=\bar Q \bar Q^\top =I$ satisfies 
$A(\Phi^\top\Phi)^{-1}A^\top = I$. Next, observe that
$$
\begin{array}{rcl}
AKA^\top& =&  \bar Q S^{1/2}U^\top K  U S^{1/2}\bar Q^\top \\
& = &  \bar Q P^\top \left[P S^{1/2}U^\top K  U S^{1/2} P^\top \right]P \bar Q^\top  \\
& = &  \bar Q P^\top  \left[ \Phi K \Phi^\top \right] P \bar Q^\top \\
& = &  \bar Q P^\top  QDQ^\top  P \bar Q^\top 
\end{array}
$$
so that choosing $\bar Q = Q^\top P$, $A =Q^\top P S^{1/2}U^\top$  satisfies  \eqref{Astrawderman}    with $\Lambda=D$, the (squared) singular values of $\Phi\Psi$. 
\end{proof}

Theorem 1 in  \cite{Strawd1978} provides, for model \eqref{LDM_OE_FIR_LM_BE_NEW},   a family of adaptive estimation rules which are \emph{minimax} w.r.t. the loss \begin{equation}\label{compoundLossD}
L(\bar\beta,\delta):=\left(\delta-\bar\beta\right)^\top D \left(\delta-\bar\beta\right)  
\end{equation}
obtained setting $Q=D$ in  \eqref{compoundLossQ}.

 We shall now attempt to provide a link between the minimax rules given in \cite{Strawd1978} and the Empirical Bayes estimator obtained plugging $\hat\lambda_{ML}$ in \eqref{EM_Coordinates}.
\begin{remark}\label{WMSE}
Note that the expected value of \eqref{compoundLossD} is exactly the weighted Mean Squared Error introduced and studied in \cite{jmlr2014} adapted to the case in which the prior covariance is not limited to being a multiple of the identity. It was shown in \cite{jmlr2014} that the marginal likelihood estimator $\hat\lambda_{ML}$ converges, as the number of data goes to infinity,  to 
\begin{equation}\label{lambda_opt}
\lambda^* : =\frac{\|\beta\|^2}{B}  = \frac{\bar\beta^\top D^{-1} \bar\beta}{B} =  \frac{g^\top  K^{-1} g}{B}
\end{equation}
where the change of coordinates in Lemma \ref{LemStraw} has been used. See also  4.2 (in particular eq. (18), (19) and Criterio 3 (page 863)) of  \cite{DeNicolao1997},  
 vs \eqref{LDM_OE_FIR_LM_BE_NEW}, \eqref{EM_Coordinates} and  \eqref{lambda_opt} of this  paper.

\end{remark}

It has been shown in \cite{jmlr2014} that the estimator $\delta^*$ obtained setting $\lambda=\lambda^*$ in \eqref{EM_Coordinates}  also minimizes (asymptotically) the expected loss $\E L(\bar\beta,\delta)$, thus
$$
\E L(\bar\beta,\delta^{*}) \leq    \E L(\bar\beta,\hat{ \bar\beta}_{LS}) 
$$
where $\hat{\bar \beta}_{LS}:=Z$ is  the least squares estimator of $\bar\beta$.  Now assume that, for $N$ large, the marginal likelihood estimator is close to its limit, i.e. 
$$
 \hat\lambda_{ML}\simeq \frac{\bar\beta^\top D^{-1} \bar\beta}{B}  \simeq \frac{Z^\top D^{-1} Z}{B} = \frac{\|Z\|^2_{D^{-1}}}{B}
$$
so  that it makes sense to approximate the empirical Bayes estimator of $\bar\beta$ replacing $\hat\lambda_{ML}$ with $\lambda =  \frac{Z^\top D^{-1} Z}{B}$, leading to 
\begin{equation}\label{EB_Strawderman}
\begin{array}{rcl}
\hat{\bar\beta}_{EB,i}& = & \left(1-\frac{\sigma^2/d_i}{ \hat\lambda_{ML} +  \sigma^2/d_i} \right) Z_i \\
& \simeq& \left(1-\frac{B \sigma^2/d_i}{ \|Z\|^2_{D^{-1}} + B  \sigma^2/d_i} \right) Z_i
\end{array}
\end{equation} 
which has a form similar to equation (2.2) in \cite{Strawd1978}. This discussion does not provide a  formal proof that the estimator $\hat{\bar\beta}_{EB}$ is minimax 
with respect to the loss \eqref{compoundLossD}, yet asymptotically it does 
dominate the least squares estimator as demonstrated in   \cite{jmlr2014}.

The discussion above  provides a clear connection between  compound estimation problems studied in \cite{Stein1961} and system identification. Indeed when performing System Identification,  there is often no interest in the single basis expansion coefficients $\alpha_i$ in \eqref{BasisExpansion} or the impulse  response coefficients $g_i$, but rather in their compound effect on output prediction error  $\tilde Y:=\Phi(g-\hat g) = \Phi A^{-1}(\bar\beta -\delta)$.  

\begin{remark}
The reader may argue that, sometimes, the interest is in the prediction w.r.t. a different input as that used for identification (e.g. low pass vs. white noise).
This fact could be accounted for with an appropriate design of the kernel matrix $K$, see eq. \eqref{weightedMSE} below. 
\end{remark}

 Expressing  the cost \eqref{compoundLossD}  in terms of output prediction error $\tilde Y$ as:
\begin{equation}\label{weightedMSE}
\begin{array}{rcl}
L(\bar\beta,\delta) & = &  L(\Phi g, \Phi \hat g) \\
& =& \tilde Y^\top \left(\Phi K \Phi^\top \right) \tilde Y = \tilde g^\top\left( \Phi^\top \Phi K \Phi^\top\Phi\right) \tilde g
\end{array}
\end{equation}
it is seen that the loss \eqref{compoundLossD} is a weighted output prediction error or, equivalently, a weighted $\ell_2$ error on the impulse response fit. The weighting $\Phi K \Phi^\top  = \Phi\Psi\Psi^\top \Phi^\top$ carries information about:
\begin{enumerate}
\item  the ``model class'' specified through the choice of the basis vectors  $\psi_i$ (and thus the covariance matrix $K = \Psi\Psi^\top$) 
\item how these vectors are mapped in output prediction through $\Phi$, which of course depends on the excitation properties since the regressor matrix $\Phi$ contains  the input $u(t)$.  
\end{enumerate}
This shows that the Bayes estimator equipped with the Marginal Likelihood estimator $\hat \lambda_{ML}$ attempts to minimize output prediction error along directions relevant to the model class considered (through the basis functions $\psi_i$) which are also  properly ``excited'' by the input. As discussed in \cite{jmlr2014}  the weighted mean squared error \eqref{weightedMSE}  gives more emphasis to directions in the ``output'' space which have high signal to noise ratio. As a result the empirical Bayes estimator    $\hat\lambda_{EB}$, which (asymptotically) minimizes the weighted MSE \eqref{weightedMSE}, is robust w.r.t. noise in the measurements $Y$. We shall elaborate upon the consequences of this observation in the next subsection. 
\subsection{Robustness of the Marginal Likelihood Estimator}
%This feature of the Marginal Likelihood estimator has been exploited in \cite{jmlr2014} where it is shown that it results in an increased robustness w.r.t. noise. 
An important implication of this robustness has been discussed in Section 6 of  \cite{PCAuto2015} in connection with the concept of  \emph{excess degrees of freedom}. We recall here that the  degrees of freedom of an estimator $\delta(Y)$ are defined (under Gaussian assumptions)  in terms  of the matrix of partial derivatives, see Remark (6.3) in \cite{PCAuto2015},
\begin{equation}\label{eq:df}
\left[D_f(\delta)\right]_{ij}:=\E\left[ \frac{\partial \left(\delta(Y)\right)_i }{\partial Y_j}\right]
\end{equation}
$D_f(\delta)$ above  is related to the so called \emph{optimism} (\cite{Hastie09}), i.e. the average amount by which the ``fitting error'' underestimates the prediction error of the estimator on new data. 

For instance, with reference to  model \eqref{LDM_OE_FIR_LM}  with $\hat Y:=\Phi\hat g=\delta(Y)$, $Y$ the ``identification data'' and $Y_{test}$ some ``fresh'' test data, we have
$$
\E \|  Y_{test}- \Phi \hat g\|^2 = \E \|  Y - \Phi \hat g\|^2 + 2\sigma^2 Tr\{D_f(\delta)\}
$$
Since the estimator $\delta(Y) = \delta_{\hat\lambda(Y)}(Y)$ depends on the data also through the estimator $\hat\lambda(Y)$  of the ``regularization'' parameter $\lambda \in \R^k$, the degrees of freedom can be expressed as:
\begin{equation}\label{CovSteinDF}
\begin{array}{rcl}
D_f(\delta) &=& \E\left[\frac{\partial \delta(Y)}{\partial Y}\right]  + \sum_{i=1}^k\E\left[\frac{\partial  \delta(Y)}{\partial \lambda_i}\frac{\partial \hat\lambda_i(Y)}{\partial Y}\right]
% \\
%& = & \E\left[\frac{\partial \delta_{\lambda}(Y)}{\partial Y}\right] + EXD_f(\delta,Y)
\end{array}
\end{equation}
The trace of the  rightmost term has been called \emph{excess degrees of freedom} and accounts for  the sensitivity of  $\hat\lambda$ w.r.t. the data $Y$. 
 Connection between optimism and sensitivities of a generic estimator is also studied  in \cite{Dinuzzoetal2007}, see Thm 2, page 2478. 

 Recall also that, for $\lambda$ fixed,  the a posteriori variance
 of $\hat Y:=\Phi\hat g= \delta(Y)$ is  related to the degrees of freedom by:
\begin{equation}\label{PosteriorVariance}
\begin{array}{rcl}
{\rm Tr}\left[  Var\{\Phi (g-\hat g)\}\right]& = &  \sigma^2{\rm Tr}\left( \E\left[\frac{\partial \delta(Y)}{\partial Y}\right]  \right) % \\
%& = &
= {\rm Tr}\left[  D_f(\delta) \right] \end{array}
\end{equation}
%Applying the variance  decomposition
%$$
%\begin{array}{rcl}
%Var\{X|\hat \lambda\} &=&\E Var\{X|\lambda,\hat\lambda \} + Var\{E[X|\lambda,\hat\lambda]|\hat\lambda\}\\
%& = & \E Var\{X|\lambda \} + Var\{E[X|\lambda]|\hat\lambda\}
%\end{array}
%$$

However the posterior variance of $\Phi \hat g$ when $\lambda$ is estimated from data is larger than  \eqref{PosteriorVariance}; in fact, also   the  uncertainty of $\hat\lambda$ should be accounted for as it (negatively) affects the posterior variance of the predictor. Thereofore it is also  desirable to control  the variability of $\hat \lambda$ or, equivalently, the excess degrees of freedom. This is, to some extent, guaranteed by the robustness of the Marginal Likelihood estimator.

% would be given 
% 
% 
%Considering the expansion 
%$$
%\delta_{\hat \lambda(Y)}(Y) \simeq \delta_{ \lambda}(Y) +\sum_{i=1}^p \left[\frac{\partial  \delta_{\lambda}(Y)}{\partial \lambda_i}\frac{\partial \hat\lambda_i(Y)}{\partial Y} E\right] 
%$$
%and computing the (posterior) variance of $\hat Y:=\Phi\hat g=\Phi \delta_{\hat \lambda(Y)}(Y)$ using Stein's Lemma one obtains:
%$$
%\begin{array}{rcl}
%{\rm Tr}\left[  Var\{ \Phi \hat g\}\right]& \simeq &  \sigma^2 {\rm Tr}\left[  D_f(\delta_{\lambda}(Y)) \right]  + \\
%& &+ {\rm cov}\left( \sum_{i=1}^p \left[\frac{\partial  \delta_{\lambda}(Y)}{\partial \lambda_i}\frac{\partial \hat\lambda_i(Y)}{\partial Y}\right], Y\right)
%\end{array}
%$$
%whee, clearly, the second term on the right hand side is related to the excess degrees of freedom in \eqref{CovSteinDF}, suggesting that keeping the excess degrees of freedom small and as little insensitive to the data as possible has a beneficial effect in terms of a posteriori variance. 

\vspace{-2mm}
\section{ Priors for dynamical systems:  the present}\label{KernelDesign}
\vspace{-2mm}

In Section \ref{CompoundEstimationID} we started by postulating that the impulse response $g$ is represented using the series expansion \eqref{BasisExpansion}. This representation
  provides a  connection with the extensive literature on ridge regression and minimax estimation in linear models. It has also been observed that  \eqref{BasisExpansion}  connects with Bayesian priors   \eqref{GaussianPrior}  provided $K:= \sum_{i=1}\psi_i\psi_i^\top$. This latter decomposition can be, for instance, obtained from the Karhunen-Lo\`eve (SVD) decomposition of the kernel $K$. Thus, designing (or selecting) the basis functions $\psi_i$ is equivalent to selecting the prior covariance matrix $K$. The recent literature
  in System Identification has extensively studied this problem from the Bayesian/Regularization point of view.  The reader is also referred to \cite{GiapiAHSS2015} for connections with the so-called 
  \emph{atomic norm}. 
  %We believe it is fair to say that  selecting the matrix $K$ is perhaps one of the most explored issues in this research area.  

In most of this paper we have drawn connections between the system identification problem \eqref{LDM_OE} and estimation in the linear model \eqref{LDM_OE_FIR_LM}. There are however important features of the unknown vector $g$ in \eqref{LDM_OE_FIR_LM} which we may like to account for when performing identification. For instance, starting from \cite{Akaike1979}, (see also Section \ref{History_NP} for a thorough discussion) ``priors'' have been introduced which account for the fact that $g$ should be a decreasing sequence (for instance in \cite{Akaike1979} it is guaranteed that    ${\rm lim}_{T\rightarrow\infty}\E \sum_{k=1}^T |g_k|^2 <\infty$.). 
In this Section we shall provide an overview of  the recent developments along this direction.  The reader is also referred to \cite{Dinuzzo2012} for further discussions. 

\subsection{Priors for SISO systems}\label{SSpline}

We shall now restrict to the simplest case $p=m=1$. The MIMO case shall be discussed later on. 
In  the seminal paper \cite{SS2010} the so called \emph{Stable Spline} prior has been introduced. The basic building block is the family of  spline priors (\cite{Wahba1990}, widely used in statistics and machine learning (\cite{Rasmussen}); unfortunately these  are not suitable  to describe impulse responses of stable linear systems which are expected to be (exponentially) decaying sequences.  The main novel  idea, expanding upon a similar construction in \cite{Neve2007} (see e.g. page 1140),  was to introduce an exponential time-warping 
\begin{equation}\label{eq:decay}
\tau = e^{-\beta t}\quad \quad \beta >0, \quad t \in \R^+
\end{equation} on the family of spline priors  so as to guarantee that realizations from the prior are, with probability one, exponentially decaying functions and thus represent BIBO stable linear systems. This reminds of the prior models discussed in Section \ref{History_NP} where different forms of decay rates for the variance of the coefficients $g_i$ were postulated. It is fair to say that \cite{Goodwin1992} is probably the most closely related; in fact   it was suggested that also the decay rate $\beta$ ($\rho$ in equation \eqref{KernelSE}) could be estimated from data using the marginal likelihood. 

This is the way Stable Spline Kernels where introduced in \cite{SS2010}; From an historical perspective, it is fair to say that the Stable Spline kernels are indeed mixing ides from \cite{Shiller1973}, where \emph{time-domain} smoothness conditions where advocated for the impulse response, and \cite{Akaike1979} where  a \emph{frequency-domain} smoothness penalty was introduced. In fact, as also discussed in \cite{Goodwin1992}, the exponential decay of the impulse response variances, implies a  Lipschitz condition on the frequency response function 
\eqref{frequency_domain_smoothness}.

The simplest version of these kernels is the so-called ``first-order'' stable spline prior, known also as Tuned-Correlated (TC) kernel, see \cite{ChenOL12}. It   can be seen as the covariance sequence of the  sampled,   time-warped,  Wiener process, so that the  matrix $K$ has entries
\begin{equation}\label{TC}
\left[K_{\lambda,\beta}\right]_{i,j} = \E g_i g_j =\lambda  {\rm min}(e^{-\beta i T_c}, e^{-\beta jT_c}) 
\end{equation} where $T_c$ is the sampling time of the discrete time system.

Using the Karunnen-Loeve decomposition (random Fourier series)  of the Wiener process (see e.g. \cite{AndersonBB1997}, eq. (11)) it is a simple calculation to see that $g$ admits the series expansion
\begin{equation}\label{Wiener-Fourier}
\begin{array}{rcl}
g_k  &=& \alpha_0 e^{-\beta k T_c} + \sqrt{2}\sum_{i=1}^\infty  \alpha_i \frac{{\rm sin}(\pi i e^{-\beta k T_c})  
}{\pi i} %& = & \psi_0 \tau + b(\tau)
\\\alpha_k &\sim  &   {\cal N}(0,\lambda) \quad i.i.d.  \end{array}
\end{equation}
which directly connects with \eqref{BasisExpansion} provided one defines
\begin{equation}\label{TC-basis}
\begin{array}{rcl}
\psi_0 &=&[1\; e^{-\beta  T_c}\; e^{-\beta 2 T_c}\; \dots e^{-\beta (T-1) T_c}]^\top\\
\psi_i & = & \sqrt{2} \left[\frac{{\rm sin}(\pi i)  
}{\pi i}  \; \frac{{\rm sin}(\pi i e^{-\beta  T_c})  
}{\pi i} \; \dots \frac{{\rm sin}(\pi i e^{-\beta (T-1) T_c})  
}{\pi i} \right]^\top
\end{array}
\end{equation}
%Note that the elements of $\psi_i$ tend to zero as $i\rightarrow\infty$ so that this series can be truncated to a finite number of terms $m$ as in \cite{
Note that the vectors $\psi_i$ are not orthogonal (the time-warping \eqref{eq:decay} destroys the orthogonality of the Karhunen-Loeve basis)  and indeed working directly with the 
SVD of the matrix \eqref{TC} could be advantageous also from a computational point of view (\cite{CarliIFAC12,TLImpAlg2013}). Yet the expansion \eqref{Wiener-Fourier} with basis \eqref{TC-basis} shows that the impulse response $g$ is modeled by the Stable Spline prior as a decaying exponential $e^{-\beta k T_c} $ ($\psi_0$) plus a random perturbation which is a time-warped Brownian Bridge (see \cite{CLPCDC2014}). As discussed in \cite{CLPCDC2014}, also  the time-warped Brownian Bridge has a local behavior (as $k\rightarrow +\infty$) of the form 
$e^{-\beta kTc}$, which derives from the fact that $sin(x) \simeq x$ for $x\rightarrow 0$. This suggests that TC kernels, and possibly their conic combinations  (\cite{ChenTAC2014}, \cite{GiapiAHSS2015}),
are well suited to describe impulse responses of finite dimensional linear systems.

There is an interesting maximum entropy interpretation of the  stable spline kernels, which can be derived under suitable  constraints on the exponential decay of derivatives of $g$ (in the continuous time formulation), see \cite{DNMTNS, ChenMaxEnt2015} and references therein.  This very same property can be easily obtained observing that the Stable-Spline Kernels can be thought of as the covariances of time-varying backward AR processes. For instance for the TC/Ctable spline of order 1 kernel \eqref{TC}, it is an easy exercise to show that it can be obtained, on any finite interval $k\in [0,T]$, as the covariance  of the AR(1) process (assuming w.l.o.g. $T_c=1$)
\begin{equation}\label{ARInterpretation}
\begin{array}{rcll}
g_T & \simeq & {\cal N}(0,c\lambda^{T}) & g_T \perp n_i \quad \forall i\\
g_{k-1} &=& g_{k} + \sqrt{c \lambda^{k-1}(1-\lambda)} n_k &  n_k \simeq {\cal N}(0,1) \quad i.i.d. 
\end{array}
\end{equation} 

This simple observation allows to derive its Maximum Entropy properties as well as the band structure of the inverse covariance directly from well known results on AR processes (\cite{GrnSzego}).

Several variations of \eqref{TC} have been proposed, including higher order splines (\cite{SS2010}) and their filtered versions (\cite{SS2011}). \cite{ChenOL12} also introduces other possibilities (notably the so called DC kernels) while later papers  (\cite{ChenTAC2014,CLPCDC2014})  investigate the potential advantages, both from the point of computational efficiency as well as in terms of their flexibility, of building multiple kernels combining ``elementary'' kernels (TC, Stable Splines, DC etc.).

A common feature of all these kernels is that  they all depend upon at least two  hyperparameters ($\beta$ and $\lambda$) and very often more  are needed. Let us say that the hyperparameter vector is partitioned in $\lambda,\xi$.  The first is a ``scale'' parameter (the variance of the coefficients $\alpha_i$) which, as we shall see later, can  also be used to perform variable selection. The second one, $\xi$, can be seen as a ``shape'' parameter vector which  controls the shape of  the basis functions $\psi_i$. For the TC kernel discussed above, $\xi = \beta$.  The analysis reported in Section \ref{CompoundEstimationID} relies on the fact that $\xi$ is fixed, while in practice it needs to be tuned (e.g. via marginal likelihood maximization \eqref{ML} or, e.g.,  cross validation). Much work is needed to fully clarify this issue, see \cite{PCAuto2015} for some results in support of the use of the marginal likelihood.

A complementary approach has been advocated in \cite{DTV_ERNSI2014,ChenLOBF2015} where the authors started directly from an orthonormal basis (e.g. the Laguerre basis, \cite{WahlbergTAC1991}) and then placed a (decaying) prior on the coefficients (see eqs. (23) and (24) in \cite{ChenLOBF2015}) to guarantee convergence of the series (see \eqref{summability}); This leads to an overall Kernel $K_{\lambda,\xi}$ which depends on the parameters $\xi$ describing the basis functions (called $p$ in \cite{ChenLOBF2015}) and those describing the prior on the coefficients (called $\alpha$ in   \cite{ChenLOBF2015});   it is suggested that  $(\lambda,\xi)$ are estimated using marginal likelihood optimization. The results in \cite{ChenLOBF2015} suggest that this kernel may have some advantages w.r.t. the TC/Stable Spline kernels discussed above.

More recent developments  (see e.g. \cite{PPCCDC2014}) aim at building Kernels which favour estimators $\hat g$ with low McMillan degree. This is done imposing a suitable penalty on the Hankel matrix of the estimated impulse response.  A formulation in the  framework of Gaussian priors will be discussed in Section \ref{SparseLowRank}.

\subsection{Priors for MIMO systems}\label{MIMO}

When the input and/or output data have dimension larger than one (i.e. $p>1$ and/or $m>1$)  the impulse response coefficients $g_k$ are $p\times m$ real matrices. The simplest solution, advocated in \cite{ChiusoPAuto2012}, is to consider the impulse responses $\left(\left[g_1\right]_{ij},..,\left[g_T\right]_{ij}\right)$ from input $j$ to output $i$ as independent random vectors. Therefore, if   the matrix sequence $g_k$ has been vectorized component wise, i.e. 
$$
g = \left(\left[g_1\right]_{11}\dots \left[g_T\right]_{11}, \left[g_1\right]_{21}\dots \left[g_T\right]_{21}\dots \left[g_1\right]_{pm}\dots\left[g_T\right]_{pm}\right)^\top
$$  the prior covariance $K$ has a block diagonal structure 
\begin{equation}\label{Dynamic-Network-Kernel}
K_{\lambda,\beta} = {\rm block\;diag}\left\{K_{\lambda_{11},\beta_{11}},....,K_{\lambda_{pm},\beta_{pm}}\right\}
\end{equation}
where 
$$
%\begin{array}{rcl}
\lambda:=[\lambda_{11},..,\lambda_{pm}] \quad \beta:=[\beta_{11},..,\beta_{pm}]
%\end{array}
$$
Depending on the problem one may assume that $\beta_j=\beta_{i}$, $\forall i,j$ so that all kernels share the same decay rate. Instead the scale factors $\lambda_i$ are typically all different and this also serves to the purpose of performing variable selection, see Section \ref{SparseLowRank} for more details. 

Additional issues concerning ``structure'', such as low McMillan degree, presence of latent variables and so on, will be discussed in the next section.

\vspace{-2mm}
\section{Structure in Modern System Identification and Sparsity}\label{SparseLowRank}
\vspace{-2mm}

In addition to stability of the impulse response $g$, in System Identification one is often interested in favoring (and discovering) a certain type of hidden structure in the identified model, the more so as the number of  variables (and thus model ``size'') grows w.r.t. the number of available data. This structure may have to do with model complexity (measured by McMillan degree), the presence of latent (unmeasurable) factors (see e.g. \cite{ForniHLR2004}), or in terms of Granger causality structure (see \cite{Granger}) which as we shall see in Section \ref{VS} is related to sparsity;  we refer to these general properties as \emph{structure}.

In the traditional ``parametric'' framework system identification   has been performed  as the cascade of model (and hence structure) selection followed by parameter estimation. 
Regularization approaches permit to perform this structure selection directly as part of the ``estimation'' process. The unknown $g$ is embedded in a ``high dimensional'' space and, eventually, the estimator $\hat g$ is pushed towards a low dimensional manifold which is to be learned from data.  This manifold may be that of low McMillan degree models,  sparse impulse response matrices, or  models admitting latent variables representations and so on.  As it will become clear finding this underlying structure from data can be formulated as a sparsity problem.

\subsection{ Variable Selection in System Identification}\label{VS}
Consider a MIMO  OE system \eqref{LDM_OE} with transfer matrix\footnote{The same type of considerations hold for a general multivariable ARMAX model \eqref{LDM}, see \cite{ChiusoPAuto2012}.} $G(q)$;   the matrix $G(q)$ may have so zero entries (meaning that $G_{ij}(q)=0$), so that the $i-th$ output does not depend on the $j-th$ input. Detecting the  zero entries of $G(q)$ from data corresponds to  performing variable selection, or equivalently finding the structure  of a causality graph which describes the ``interaction'' between inputs and outputs; this problem has been framed in the context of sparsity in \cite{MaterassiIG2009,ChiusoPCDC2010,VandebergheJMLR2010,Vincent2011}, where the link with literature on sparse linear models (\cite{Tibshirani94,LARS2004}) and compressive sensing (\cite{Donoho2006})  has been exploited. Among the papers mentioned above, only \cite{ChiusoPCDC2010} explicitly introduces regularization  on the impulse responses to avoid overfitting when the number of lags $T$ needs to be large. Further developments can be found in \cite{ChiusoPNIPS2010,ChiusoPAuto2012} where Bayesian techniques have been proposed. In fact this ``Bayesian'' approach (which corresponds to using the kernel \eqref{Dynamic-Network-Kernel} in conjunction with marginal likelihood estimators for $\beta$ and $\lambda$) is related to so-called Automatic Relevance Determination (ARD, 
 \cite{McKayARD}) and Sparse Bayesian Learning (SBL, \cite{Tipping2001}), well known in the Machine Learning community. In fact, as shown in \cite{ChiusoPAuto2012,ChenTAC2014}, the marginal likelihood  estimators $\hat\lambda_{ij}$ of $\lambda_{ij}$ in \eqref{Dynamic-Network-Kernel} are ``threshold estimators'' meaning that $\hat\lambda_{ij}=0$ on a set of data of positive measure.  
 
 \begin{example}\label{example_simple}
 A very simple example is the following:
 assume data $y_i$, $i=1,...,N$ are generated according to $$y_i = \theta  + e_i \quad \theta \in \R$$
 where $e_i$ are i.i.d. zero mean Gaussian variables with known variance $\sigma^2$. Assuming a prior $\theta \sim {\cal N}(0,\lambda)$, it is an easy exercise to show that the maximum marginal likelihood estimator $
 \hat 
\lambda_{ML}$ \eqref{ML} is given by (see also eq. \eqref{MLStein}):
$$
\hat \lambda_{ML} = {\rm max}\left(0,\left(\frac{1}{N}\sum_{i=1}^N y_i\right)^2 - \sigma^2\right)
$$
showing that $\hat\lambda_{ML}$ is non negative only provided the squared sample average  $\left(\frac{1}{N}\sum_{i=1}^N y_i\right)^2$ of $y$ is above the noise variance  $\sigma^2$. 
\end{example}

 When $\hat\lambda_{ij}=0$ the prior for $\left(\left[g_1\right]_{ij}\dots \left[g_T\right]_{ij}\right)$, the impulse response from input $j$ to output $i$,  has zero mean and zero variance resulting in a posterior estimator $\E_{\hat\lambda}[ \left[g_k\right]_{ij}|Y]=0$, $\forall k=1,..,T$. There is evidence (see e.g. \cite{ChiusoPAuto2012} in the context of dynamical systems) that these Bayesian approaches are superior to ``Lasso-type'' procedures, and in fact an extensive literature has attempted with some success to show this, see \cite{Wipf_IEEE_TIT_2011,jmlr2014} and references therein. In particular \cite{jmlr2014} shows that these Bayesian procedures have a better ``shrinking vs. sparsity'' tradeoff than ``Lasso-type'' procedures (such as Group Lasso).

 \subsection{Low McMillan degree, Hankel matrices and sparsity}

 A major limitation of a Gaussian prior with block diagonal covariance   $K$ as in \eqref{Dynamic-Network-Kernel} 
 is that this prior does not capture the inherent relations between the transfer functions $G_{ij}(q)$ when $G(q)$ has low McMillan degree.  This is particularly important if $m$ and $p$ are large.
%but the underlying transfer function $G$ has low McMillan degree (say $n$), so that $g_k = CA^{k-1}B$ for suitable matrices $C\in \R^{p\times n}$, $A\in \R^{n\times n}$ and 
%$B\in \R^{n\times m}$.  
Several authors have addressed this issue. For instance \cite{Fazel01} 
 tackles this problem replacing $J_\eta(g)$ in  \eqref{RegLS} with a nuclear norm  penalty,
\begin{equation}\label{eq:NN}
J_\eta(g):= \eta \|{\cal H}(g)\|_* = \eta\, Trace \sqrt{{\cal H}(g){\cal H}(g)^\top}
\end{equation}
where ${\cal H}(g)$ is the block Hankel matrix built with the impulse response sequence $g$. The reason for using the nuclear norm stems from Theorem 1 in  \cite{Fazel01} which  shows that the nuclear norm $\|A\|_*$  is the convex envelope of the rank of $A$ on the ball $\|A\|<1$. Since  $J_\eta(g)$ approximates the rank function it is expected that solutions $\hat g$ of  \eqref{RegLS}  with low nuclear norm $ \|{\cal H}(\hat g)\|_*$ are close to having ${\cal H}(\hat g)$ of low rank and thus low McMillan degree of $\hat g$.  Tens of papers have been published on the use of the nuclear norm in System Identification, and it is certainly impossible to survey them here. We shall only mention  \cite{Shah2012} which discusses the link between Hankel nuclear norm and \emph{atomic} norm, \cite{HWR_SYSID_2012} which applies these ideas to Box-Jenkins modeling,  \cite{Chiuso13} where Nuclear Norm-based techniques are first compared in a simulation study to other regularization methods,  \cite{LiuVSIAM2009}, \cite{AyazogluS12} and \cite{Wahlberg2015} where  computational issues were addressed, and \cite{Verhaegen2015}  which applies nuclear norm penalties  to subspace-type algorithms. See also  \cite{PPCSYSID2015} for an overview. 

Unfortunately, using the  nuclear norm penalty is known to introduce a (possibly significant) bias in the estimator \eqref{RegLS}. This can be easily explained by the fact that the nuclear norm is really an $\ell_1$-norm on the singular values of the Hankel matrix, and thus related to Lasso-type approaches. Lasso is well known to introduce significant bias on large coefficients. This is the reason why variations of the Lasso (such as the Adaptive Lasso \cite{AdaptiveLasso} or SCAD \cite{FanLiSCAD_JASA2001} and so on) have been introduced in the Statistics literature. 

 Inspired  by the analysis of the ``Shrinking vs. Sparsity'' tradeoff of Lasso and Bayesian procedures found in \cite{jmlr2014}  and by some recent papers on rank minimization (see \cite{Wipf12,Mohan12jmlr}), the paper  \cite{PPCCDC2014}  introduces a new ``Bayesian'' procedure for rank minimization in the context of system identification.
 
A ``general'' form (see \cite{Wipf12,PPCCDC2014}) of a low-rank promoting prior for a matrix $A$, which involves quadratic forms of $A$,  has the structure 
\begin{equation}\label{low_rank_prior}
p_\lambda(A)\propto e^{-\lambda Tr\left(Q AA^\top\right)}
\end{equation}
where $Q=Q^\top\geq 0$  plays the role of an hyperparameter. Slightly different structures have been proposed in \cite{Rojas2014_rank} where ``left'' and ``right''  hyperparameters have been considered.

Here, as in \cite{PPCCDC2014}, we consider a simplified version of \eqref{low_rank_prior} which can be derived under maximum entropy arguments (see \cite{PPCSYSID2015}),  where the matrix $Q$ takes on  a very simple structure
$$
Q = \lambda_1 \Pi_{U_n} + \lambda_2 \Pi_{U^\perp_n};
$$
 $\Pi_{U_n}$ and  $\Pi_{U^\perp_n}$ are, respectively, the projectors on the column spaces of $U_n$ and on   its orthogonal complement $U_n^\perp$. 
This, complemented with a ``stable spline'' prior yields to the  so called \emph{Stable-Hankel}  (SH) prior
 \begin{equation}\label{eq:SH}
 \begin{array}{rcl}
p_\eta(g)&\propto & e^{-\frac{1}{2} J_\eta(g)} \\
J_\eta(g)&:=&\lambda_s g^\top K_s^{-1} g  +  
 \lambda_1 Trace\left[\Pi_{U_n}{\cal H}(g){\cal H}(g)^\top\right] \\ & & +  \lambda_2 Trace \left[\Pi_{U^\perp_n}{\cal H}(g){\cal H}(g)^\top\right]
\end{array}
\end{equation}
where $K_s$ is a ``stable'' kernel (with decaying diagonal elements) so as to guarantee the estimators $\hat g$ are decaying functions of the lag. 
 The recent paper \cite{GiapiAHSS2015} discusses the role of the stability constraint, and the effect of not exploiting it in Hankel nuclear norm based penalties.

 The scale factors $\lambda_s$, $\lambda_1$ and $\lambda_2$ are estimated using the marginal likelihood \eqref{ML};  also in this case the marginal likelihood can be computed in closed form, being \eqref{eq:SH} a quadratic form in $g$.  Also the matrix $U_n$ along with the number of its columns $n$ can estimated using the marginal likelihood, see \cite{PPCSYSID2015} for details. Suffices here to observe that the splitting of the penalty along 
the (column) spaces of $U_n$ and its orthogonal complement $U_n^\perp$ allows to reduce the bias along the ``signal'' subspace, while being robust w.r.t. ``errors'' within ${col\;span}\{U_n\}$. This research direction is still in its infancy and much work is needed to fully clarify the potential and limitations of this approach. The interested reader is referred to \cite{PPCSYSID2015} which has  been recently  presented at the 2015 IFAC symposium on System Identification.

\subsection{Latent variable models and sparsity}

When dealing with  high  dimensional data, and thus predictors depending upon a large number of variables, there is a need for parsimonious models  which are at  the same time robust (w.r.t. noise in the data) and interpretable. As an example consider the linear static model 
$$y = M x + e
$$ where the data $y$ are expressed as a linear function $Mx$ of some  variable $x$ (measured or unmeasured).
 Two common alternatives explored in the literature to control the complexity of $M$ is to either assume it is of \emph{low rank} ($M=M_{LR}$) or it is \emph{sparse} ($M=M_S$), or combinations thereof ($M=M_{LR}+M_S$, see e.g. \cite{Parillo_SLR2011} and references therein for the use of this decomposition in a static scenario). 
 
 The low rank component $M_{LR}$ models the presence of few underlying ``common'' variables which  capture the \emph{co-movement} of the observables $y$. Relevant keywords here are Principal Component Regression and Factor Models. On the other hand a sparse model can be encoded in terms of a network of (conditional) dependences (\cite{LAURITZEN_1996}). 
 
 Such models can  also be studied in a dynamic scenario where $y$ is a multivariate stochastic process and the matrix  $M$ is a replaced by a transfer matrix $G(q)$. For instance 
 the alternative between considering Principal Component Regression and Regularized Regression when the number of ``regressors'' becomes large is studied in \cite{DeMol2008}. 

However a through study of ``Sparse + Low rank'' dynamical models seems to be lacking; the paper \cite{ZorziS2014} discusses Sparse plus Low rank AR models 
where sparsity is related to smoothing conditions in the style of \cite{VandebergheJMLR2010}.  If instead causal (\emph{a la Granger}) conditional dependences are considered, 
some preliminary work in a Bayesian framework can be found in 
  (\cite{ZorziC2015}) where the output process $y(t) \in \R^p$ is modeled as  
\begin{equation}\label{eq:SLR}
\begin{array}{rcl}
y(t) &=& F x(t) + S(q) y(t) + v(t)  \\
x(t) &=&  H(q) y(t) + w(t)
\end{array}
\end{equation}
where $x(t) \in \R^r$, $r<<p$ is a latent process (not measurable), $S(q)$ and $H(q)$ are  strictly causal  stable  transfer matrices and $S(q)$ is sparse in the sense that many of its entries $S_{ij}(q)$ are null transfer functions. 
Model \eqref{eq:SLR} can be given a graphical interpretation with the components of $y(t)$ being ``visible'' nodes, those of $x(t)$ being hidden nodes and edges encoding the (Granger causality) dependences between these time series. The Bayesian framework discussed in this paper can be used to simultaneously infer the transfer matrices $S(q)$, $H(q)$, the matrix $F$ as well as the number of latent variables $x(t)$. Details can be found in \cite{ZorziC2015}.

%CITARE 
%@TECHREPORT{Knox2001,
%  AUTHOR =       {T. Knox and {J.H.} Stock and {M.W.} Watson},
%  TITLE =        {Empirical {B}ayes Forecast of one time series using many predictors},
%  INSTITUTION =  {National Bureau of Economic Research},
%  YEAR =         {2001}
%  %note =         {Sumbitted to IEEE Trans. on Aut. Control, available at: http://www.dei.unipd.it/{\tiny{$\sim$}}chiuso},
%}
%
%@ARTICLE{GeorgreFoster_Biometrika_2000,
%  AUTHOR =       {{E.I.} George and {D.P.} Foster},
%  TITLE =        {Calibration and Empirical {B}ayes variable selection},
%  JOURNAL =      {Biometrika},
%  YEAR =         {2000},
%  volume =       {87},
%  number =       {4},
%  pages =        {731-747}
%}

%\subsection{Optimization issues}
\vspace{-2mm}
\section{Conclusions and Future Outlook}\label{Outlook}
\vspace{-2mm}

As briefly outlined in Sections \ref{KernelDesign} and \ref{SparseLowRank} the recent literature has addressed the problem of designing Kernel structure which can capture (and actually favour) structural properties of dynamical systems.  In a sense the Bayesian formulation can be thought of as  \emph{lifting one level up} the classical problem of estimating the system parameters, to the problem of determining the parameters of a flexible model class (the hyperparameters describing the prior), see also \cite{LindleySmith1972}.\\

Of course there is a tradeoff between how well a prior density can be shaped to give low bias and how well the hyperparameters can be estimated from data (low variance), thus leading to the usual \emph{bias-variance tradeoff}, which can be faced again via a  \emph{parsimony} principle.  We believe much work remains to be done along this direction, both in terms of design as well as in terms of analysis of the tuning procedures.

In terms of kernel design, it should be apparent that the problems considered  in this paper are only a few examples (possibly the most straightforward ones) and  it is indeed to be expected that many others can be considered and framed with the Bayesian language. Yet the reader should keep in mind that the System Identification problem should not be regarded as a general function regression (or interpolation) problem. There are intrinsic properties of dynamical systems that should always be accounted for and used when designing prior models or, equivalently, model classes. 

As the dimension of data increases, both in terms of number of time series $p$  which need to be jointly modeled, as well as in terms of number of observations $N$ one is able or should process at a time, the computational aspect becomes more and more relevant. The more so if real time use of these modeling tools is envisaged. 
For instance in \cite{BonettiniCPSIAM2014} fast first order  methods are considered for marginal likelihood optimization. These techniques have been adapted and used in \cite{PPCSYSID2015}  showing significant improvements in terms of computational time w.r.t. off-the-shelf techniques.

Yet much work needs to be done to make the techniques discussed in this paper  scale well with $p$ and $N$; to such purpose we believe 
that optimization issues should enter at a much earlier stage (such as the design of the prior) so as to guide the developments of algorithms and methods. 
We envision that stochastic optimization and possibly early stopping ideas may turn out to be useful in this context (see e.g. \cite{RosascoTV2014}).

We also believe that  further work is needed regarding  the use,   design and implementation of MCMC methods for handling either more complex situations, e.g. where  the marginal likelihood cannot be computed in closed form \cite{PilTAC2011}, or  for computing so-called ``full-Bayes'' solutions and 
accurately describing the uncertainty in the estimators  (\cite{NeveDNM2008}, \cite{RomeresCDC2015}), not necessarily properly quantified by the empirical Bayes estimators,  \cite{PCAuto2015}.
\vspace{-2mm}
\section*{Acknowledgements}
\vspace{-2mm}
The author is grateful to a number of colleagues with whom he has shared ideas, thoughts and joint research efforts in the past years as well as for providing comments on early versions of this manuscript. Above all my colleague Gianluigi Pillonetto,
coauthor of a long list of papers, but also, in alphabetical order, Tianshi Chen, Manfred Deistler, Giuseppe De Nicolao, Michel Gevers, Lennart Ljung, Giorgio Picci, Victor Solo, Roberto Tempo, Mattia Zorzi and my present students Giulia Prando and Diego Romeres. 

%\bibliography{References}             % bib file to produce the bibliography
                                                     % with bibtex (preferred)

%\appendix
%\section{IF NEEDED 1}    % Each appendix must have a short title.
%\section{IF NEEDED 2}              % Sections and subsections are supported  
%        

                                                                 % in the appendices.
\end{document}